\documentclass[journal]{IEEEtran}

\usepackage{amssymb,amsmath,graphicx,paralist,subfigure,pdfpages,cite,amsthm,wasysym,algorithmic,algorithm}
\newtheorem{lem}{Lemma} 
\newtheorem{thm}{Theorem}
\newtheorem{definition}{Definition}

\def\ln{{\rm ln}}

\def\mc{\mathcal}

\def\mb{\mathbf}
\def\mbb{\mathbb}
\def\ra{\rightarrow}
\IEEEoverridecommandlockouts

\begin{document}
\title{\begin{color}{black}Localization in mobile networks via\\ virtual convex hulls\end{color}}
\author{Sam Safavi,~\emph{Student Member,~IEEE}, and Usman A. Khan,~\emph{Senior Member,~IEEE}\thanks{The authors are with the Department of Electrical and Computer Engineering, Tufts University, 161 College Ave, Medford, MA 02155, {\texttt{\{sam.safavi@,khan@ece.\}tufts.edu}}. This work is partially supported by an NSF Career award: CCF \# 1350264.}}
\maketitle
\thispagestyle{empty}

\begin{abstract}
In this paper, we develop a \textit{distributed} algorithm to localize an arbitrary number of agents moving in a bounded region of interest. We assume that the network contains \textit{at least one} agent with known location (hereinafter referred to as an anchor), and each agent measures a noisy version of its motion and the distances to the nearby agents. We provide a~\emph{geometric approach}, which allows each agent to: (i) continually update the distances to the locations where it has exchanged information with the other nodes in the past; and (ii) measure the distance between a neighbor and any such locations. Based on this approach, we provide a \emph{linear update} to find the locations of an arbitrary number of mobile agents when they follow some convexity in their deployment and motion. 

Since the agents are mobile, they may not be able to find nearby nodes (agents and/or anchors) to implement a distributed algorithm. To address this issue, we introduce the notion of a \emph{virtual convex hull} with the help of the aforementioned geometric approach. In particular, each agent keeps track of a virtual convex hull of other nodes, which may not physically exist, and updates its location with respect to its neighbors in the virtual hull. We show that the corresponding localization algorithm, in the absence of noise, can be abstracted as a Linear Time-Varying (LTV) system, with non-deterministic system matrices, which asymptotically tracks the true locations of the agents. We provide simulations to verify the analytical results and evaluate the performance of the algorithm in the presence of noise on the motion as well as on the distance measurements. 
\end{abstract}

\section{Introduction}\label{sec1}
Localization is a well-studied problem, which refers to the collection of algorithms that estimate the location of nodes in a network. Relevant applications include traffic control, industrial automation, robotics, and environment monitoring~\cite{estrin2001instrumenting,pottie2000wireless,estrin1999next,steere2000research}.
In terms of the information used for estimating the locations, localization schemes in Wireless Sensor Networks (WSN) can be classified as range-based,\textcolor{black}{\cite{Pandey2016,zhang2005range,7496598}},
and range-free,\textcolor{black}{\cite{hu2004localization,dil2006range,rudafshani2007localization,stevens2007dual}}. 
While the former depends on measuring the distance and/or angle between the nodes, 
	the latter makes no assumptions about the availability of such information and relies on the connectivity of the network.
The literature on localization also consists of centralized and distributed approaches. Centralized algorithms,~\cite{patwari2003relative,thomas2005revisiting,roweis2000nonlinear,aspnes2006theory}, despite their benefits, are impractical in large networks where each node has limited power and communication capability. 
Some notable distributed localization techniques include successive refinements,~\cite{albowicz2001recursive,savarese2001location,vcapkun2002gps},
multilateration,~\cite{sriv:02,nagpal2003organizing,savvides2001dynamic}, \textcolor{black}{estimation based methods, multidimentional scaling~\cite{ji2004sensor}}, and graph-theoretical methods,~\cite{anderson2009graphical,deghat2011distributed,lederer2009connectivity}. 
%

When ranging data is noisy, the following estimation-based localization techniques are widely used. Maximum Likelihood Estimation (MLE) methods estimate the most probable location of a mobile node based on prior statistical models. MLE only requires (current) measured data and does not integrate prior data. Sequential Bayesian Estimation (SBE) methods use the recursive Bayes rule to estimate the likelihood of an agent's location. The solution to SBE is generally intractable and cannot be determined analytically. An alternative approach is Kalman-based techniques, which are only optimal when the uncertainties are Gaussian and the system dynamics are linear,~\cite{fox2003bayesian,amundson2009survey}. However, localization has always been considered as a nonlinear problem and hence the optimality of Kalman-based solutions are not guaranteed. To address the nonlinear nature of localization problems, other suboptimal solutions to approximate the optimal Bayesian estimation include Extended Kalman Filter (EKF) and Particle Filters (PF),~\cite{amundson2009survey}. In particular, Sequential Monte Carlo (SMC) method is a PF that exploits posterior probability to determine the future location of an agent. 

\textcolor{black}{Localization algorithms specifically designed for mobile networks have been proposed in~\cite{hu2004localization,rudafshani2007localization,stevens2007dual,dil2006range}. 
Ref.~\cite{hu2004localization} introduces the Monte Carlo Localization (MCL) method, which exploits mobility to improve the accuracy of localization. Inspired by~\cite{hu2004localization}, the authors in~\cite{rudafshani2007localization} propose Mobile and Static sensor network Localization (MSL*) that extends MCL to the case where some or all nodes are static or mobile. Ref~\cite{stevens2007dual} proposes variations of MCL, namely dual and mixture MCL to increase the accuracy of the original MCL algorithm. On the other hand, Ref~\cite{dil2006range} provides Range-based Sequential Monte Carlo Localization method (Range-based SMCL), which combines range-free and range-based information to improve localization performance in mobile sensor networks.}

In this paper, we consider localization of a mobile network, assuming that each agent measures a noisy version of its distances to the nodes (agents and/or anchors) in its communication radius\footnote{In wireless networks, the distance may be estimated with, e.g., Received Signal Strength (RSS), Time of Arrival (ToA), Time Difference of Arrival (TDoA) measurements,~\cite{PatwariThesiss}. We further note that the distance information can also be obtained using a camera at each agent,~\cite{camera}.}, and the motion it undertook, e.g., by using an accelerometer. 
We focus on scenarios where GPS (or related positioning system) is compromised and/or unavailable, and develop a collaborative algorithm for a network of agents/robots to find their locations in order to implement a location-aware task or a mission.
{\textcolor{black}{As an example, consider a network of mobile robots with no central or local coordinator and with limited communication, whose task is to transport goods in an indoor facility where GPS signals are not available. In order to perform a delivery task, each mobile robot has to know its own location first.} In such settings, we are interested in developing \textcolor{black}{an anchor-based \textit{distributed}} algorithm \textcolor{black}{that uses only the distance and angle measurements} to track the agent locations such that the convergence is \textcolor{black}{\textit{invariant to the initial position estimates}}. However, the implementation \textcolor{black}{and analysis of such distributed algorithm} in mobile networks is not straightforward, because: 
\begin{inparaenum}[(i)]
\item an agent may not be able to find nearby nodes \emph{at any given time} to implement a distributed algorithm;
\item an agent may not be in the proximity of an anchor \emph{at any given time};
\item the neighborhood at each agent is dynamic, resulting into a \textit{time-varying} distributed algorithm. 
\end{inparaenum}

In this context, the main contribution of this work is to develop a \textit{linear} framework for localization that enables us to circumvent the challenges posed by the predominant  nonlinear approaches to this problem. This linear framework is not to be interpreted as a linearization of an existing nonlinear algorithm. Instead, the nonlinearity from range to location is embedded in an alternate representation provided by the barycentric coordinates. However, forcing this linear update comes at a price, i.e., a location update is only possible when an agent lies in a triangle (convex hull in $\mathbb{R}^2$) formed by three neighboring nodes. Since such neighbors may never exist, see (i) above, we introduce the notion of a \emph{virtual convex hull}. We then show that localization in mobile networks can be achieved with \textit{exactly one anchor} when the total number of nodes (agents and anchors) is at least~$4$ in~$\mbb{R}^2$. \begin{color}{black}This is in stark contrast to the requirement of at least $3$ anchors that is prevalent in the triangulation-based localization literature. See~\cite{7438736} for a recent survey on localization approaches with single vs. multiple anchor nodes.\end{color}

As we will discuss in Section~\ref{sec5}, we abstract the corresponding localization algorithm as a Linear Time-Varying~(LTV) system whose system matrices may be: \emph{identity}--when no location update occurs; \emph{stochastic}--when there is no anchor among the neighbors of the updating agent; and, \emph{strictly sub-stochastic}--when the neighborhood includes at least one anchor. Note that the order in which the updates occur and the weights assigned to the neighbors are not known a priori. To address these challenges, we \textcolor{black}{apply} a novel method to study LTV convergence by partitioning the entire chain of system matrices into \textit{slices} and {\textcolor{black}{relating} the convergence rate to the slice lengths. In particular, we show that the algorithm converges if the slice lengths do not grow faster than a certain exponential rate. Since our localization scheme is based on the motion and the distance measurements, it is meaningful to evaluate the performance of the algorithm when these parameters are corrupted by noise. Therefore, we study the impact of noise on the convergence of the algorithm and provide modifications to counter the undesirable effects of noise. Owing to the linear representation, we are able to provide a few simple yet comprehensive modifications to the proposed approach, \textcolor{black}{which as shown in the simulations lead to a bounded error in the location estimates.}

We now describe the rest of the paper. In Section~\ref{sec2}, we formulate the problem. We introduce the geometric approach to track the distances in Section~\ref{sec3}. We then provide our localization algorithm that relies on the notion of virtual convex hull in Section~\ref{sec4}, followed by the convergence analysis in Section~\ref{sec5}. In Section~\ref{noise}, we examine the effects of noise on the algorithm. \textcolor{black}{We provide simulation results and compare the performance of the algorithm with MCL, MSL*, Dual MCL, and Range-based SMCL algorithms in Section~\ref{sec7}. Finally, Section~\ref{sec9} concludes the paper.}

\section{Preliminaries and Problem formulation}\label{sec2}
Consider a network of~$N$ mobile agents, in the set~$\Omega$, with unknown locations, {\textcolor{black}{and~$M$ anchor(s), in the set~$\kappa$, with fixed known locations}, all located in~$\mathbb{R}^{2}$; let~$\Theta=\Omega\cup\kappa$ be the set of all nodes. Let~$\mb{x}_k^{i\ast}\in\mbb{R}^2$ be a row vector that denotes the \emph{true location} of the~$i$-th agent,{\textcolor{black}{~$i\in\Omega$}, at time~$k$, where~$k\geq0$ is the discrete-time index.  Regardless of what motion modality (aerial, ground, kinematics) is employed by the agents, their motion can be expressed as the deviation between the current and next locations, i.e., 
\begin{eqnarray}\label{Eq1}
\mb{x}_{k+1}^{i\ast} = \mb{x}_{k}^{i\ast} + {\widetilde{\mb{x}}_{k+1}^i},\qquad i\in\Omega,
\end{eqnarray}
where~${\widetilde{\mb{x}}_{k}^i}$ is the true motion vector at time $k$. \textcolor{black}{We assume that the agents move along straight lines, or otherwise the motion trajectories can be approximated by a piecewise linear model. This is a common assumption in robotics literature~\cite{choset2005principles}, e.g., a simple differential wheeled robot whose wheels run in the same direction and speed, will move in a straight line.} We also assume that agent~$i$ measures a noisy version,~${\widehat{\mb{x}}_{k}^i}$, of this motion, e.g., by using an accelerometer: 
\begin{eqnarray}\label{Eq2}
{\widehat{\mb{x}}_{k}^i}={\widetilde{\mb{x}}_{k}^i}+n^i_{k},
\end{eqnarray}
where $n^i_{k}$ is the accelerometer noise at time $k$. We assume that the motion is restricted in a bounded region in~$\mbb{R}^2$. We define the distance between any two nodes, $i$ and~$j$, \textcolor{black}{measured at the time of communication,~$k$}, as {\textcolor{black}{$\widetilde{d}^{ij}_k$}. We assume that the distance measurement,~$\widehat{d}^{ij}_k$, at node $i$ is not perfect and includes noise,~i.e.,
\begin{equation}\label{Eq3}
\widehat{d}^{ij}_k = {\textcolor{black}{\widetilde{d}^{ij}_k}} + r^{ij}_k,
\end{equation}
in which $r^{ij}_{k}$ is the {\textcolor{black}{noise in} the distance measurement at time~$k$. \textit{The problem is to find the locations of the mobile agents in the set~$\Omega$ given any initialization of the underlying algorithm.} We then consider the minimal number of anchors required for such a process to work. 

{\bf Main idea:} We now describe the main idea behind our approach. In order to avoid nonlinearity in the solution\footnote{When only distances to at least~$m+1$ anchors are known \textcolor{black}{in~$\mbb{R}^m$},~trilateration, e.g., in~$\mbb{R}^2$, requires $3$ anchors and solves $3$ circle equations,~\cite{4407221}. Clearly, trilateration is nonlinear and coupled in the coordinates; while an iterative procedure built on these nonlinear updates does not converge in general.}, we consider a barycentric-based linear representation of positions; the advantage of a linear approach is that the convergence is independent of the initial location estimates. Let $\Theta_i(k)$ denote a set of three agents such that agent~$i$ lies inside {\textcolor{black}{their} convex hull,~$\mc{C}(\Theta_i(k))$,
at time~$k$. In the barycentric-representation, the location of node $i$ is described by a linear-convex combination of the nodes in~$\Theta_i(k)$, i.e.,
\begin{equation}\label{Eq5}
\mathbf{x}_k^{i\ast}=a_{k}^{{i}1} \mathbf{x}_k^{1\ast}+a_{k}^{{i}2} \mathbf{x}_k^{2\ast}+a_{k}^{{i}3} \mathbf{x}_k^{3\ast},\qquad\Theta_i(k)=\{1,2,3\}
\end{equation}
in which $\mathbf{x}_k^{i\ast}$ represents the true location of node~$i$ at time~$k$, and the coefficients, $a_{k}^{{i}j}$'s, are the \emph{barycentric coordinates}, associated to M{\"o}bius,~\cite{mobius1827barycentrische}, given by
\begin{equation}\label{Eq6}
a_{k}^{{i}j}=\frac{A_{\Theta_i(k)\cup\{i\}\setminus j}}{A_{\Theta_i(k)}},
\end{equation}
where~$A_{\Theta_i(k)\cup\{i\}\setminus j}$ denotes the area of the convex hull formed by the agents in the subscript set. To implement this update, agent~$i$ has to lie inside the convex hull of the (three) neighbors in the set~$\Theta_i(k)$. {\textcolor{black}{Appendix \ref{incl}}} provides a simple procedure,~\cite{khan2009distributed}, to test if an agent lies inside or outside of the convex hull formed by the nearby nodes. Finally, note that the barycentric coordinates are always positive and they sum to~$1$.

{\bf Challenges and our approach}: To implement the above barycentric-based procedure, agent~$i$ must acquire the mutual distances among the~$4$ nodes (in the set $\{i\cup\Theta_i(k)\}$). However, an agent may never find such a convex hull because the nodes are mobile, in a region possibly full of obstacles, with limited communication and/or visual radius. To address this challenge, we provide a geometric approach that allows an agent to track its distance to any location, where it has exchanged information with the other nodes in the past (Section~\ref{gm1}). We further show that at the time of communication an agent can compute the distance between a neighbor and any such locations (Section~\ref{gm2}). Using these methods, we introduce the notion of a virtual convex hull, which is a triangle whose vertices are located at the virtual locations where the agent exchanged information in the past (Section~\ref{sec4}). The \textcolor{black}{described approach} results in a rather opportunistic algorithm where an update occurs only when a certain set of real and virtual conditions are satisfied. The subsequent analysis is to characterize the convergence of such LTV algorithms whose system matrices are non-deterministic. In this paper, we rigorously develop the conditions of convergence and study the location-tracking performance of this procedure (Section~\ref{sec5}). 

We now enlist our assumptions:

{\bf A0:} Each anchor,~$i\in\kappa$, knows its location,{\textcolor{black}{~$\mb{x}_k^{i\ast}=\mb{x}_0^{i\ast}$}, $\forall k$.

{\bf A1:} Each agent,~$i\in\Omega$, has a noisy measurement,~$\widehat{\mb{x}}_k^i$, of its motion vector,~${\widetilde{\mb{x}}_{k}^i}$, see Eq.~\eqref{Eq2}.

{\bf A2:} Each agent,~$i \in \Omega$, has a \textcolor{black}{noisy measurement,~$\widehat{d}^{ij}_k$,} of the distances to nodes, $j \in \Theta$ within a radius,~$r$, see Eq.~\eqref{Eq3}.

Under the above assumptions, we are interested in finding the true locations of each agent in the set~$\Omega$, without the presence of any central coordinator. In the process, we study what is the minimal number of anchors required and in particular: does mobility reduce the number of anchors\footnote{The minimal number of anchors required to solve a localization problem, without ambiguity, given only distance measurements is~$m$ in~$\mbb{R}^m$,\textcolor{black}{~\cite{4407221}}.} from~$3$ in~$\mathbb{R}^2$ (Section~\ref{sec5}). In the following, we first discuss the ideal scenario when the motion and distance measurements are not \textcolor{black}{affected by} noise. We then evaluate the performance of the algorithm in the presence of {\textcolor{black}{noise}} on the motion as well as on the distance measurements, and present a modified algorithm to counter the undesirable effects of noise (Section~\ref{noise}).

\section{A Geometric Approach towards Localization}\label{sec3}
\textcolor{black}{In this section, we} consider the motion and distance in the noiseless case, i.e.,~$n_k^i=0$ and $r_k^{ij}=0$, in Eqs.~\eqref{Eq2} and~\eqref{Eq3}.
Let the \textit{true} location of the~$i$-th agent,~$i\in\Omega$ in~$\mbb{R}^2$, be decomposed as~${\mb{x}_k^{i\ast}}=[{x_k^{i\ast}}~{y_k^{i\ast}}]$. We then describe the motion of agent $i$ as:
\begin{eqnarray}\label{1}
{x_{k+1}^{i\ast}}&=&{x_{k}^{i\ast}}+d_{k\rightarrow k+1}^i\cos\left(\textcolor{black}{\sum_{{k}^{\prime}=1}^{k+1}\theta_{k^{\prime}\rightarrow k^{\prime}+1}^i}+\theta_{0}^i\right),\nonumber\\
{y_{k+1}^{i\ast}}&=&{y_{k}^{i\ast}}+d_{k\rightarrow k+1}^i\sin\left(\textcolor{black}{\sum_{{k}^{\prime}=1}^{k+1}\theta_{k^{\prime}\rightarrow k^{\prime}+1}^i}+\theta_{0}^i\right),
\end{eqnarray}
where~$d^i_{k\rightarrow k+1}$ and~$\theta^i_{k\rightarrow k+1}  $ denote the distance and angle traveled by agent~$i$, between time~$k$ and~$k+1$, and~$\theta_{0}^i$ is agent~$i$'s \textcolor{black}{initial orientation}\footnote{Note that these initial orientations,~$\theta_{0}^i$'s, are arbitrary and we \textcolor{black}{neither assume any global synchronization nor we know what the true angles are.}}, see Fig.~\ref{f0}. 
\begin{figure}[!h]
	\centering
	\includegraphics[width=1.5in]{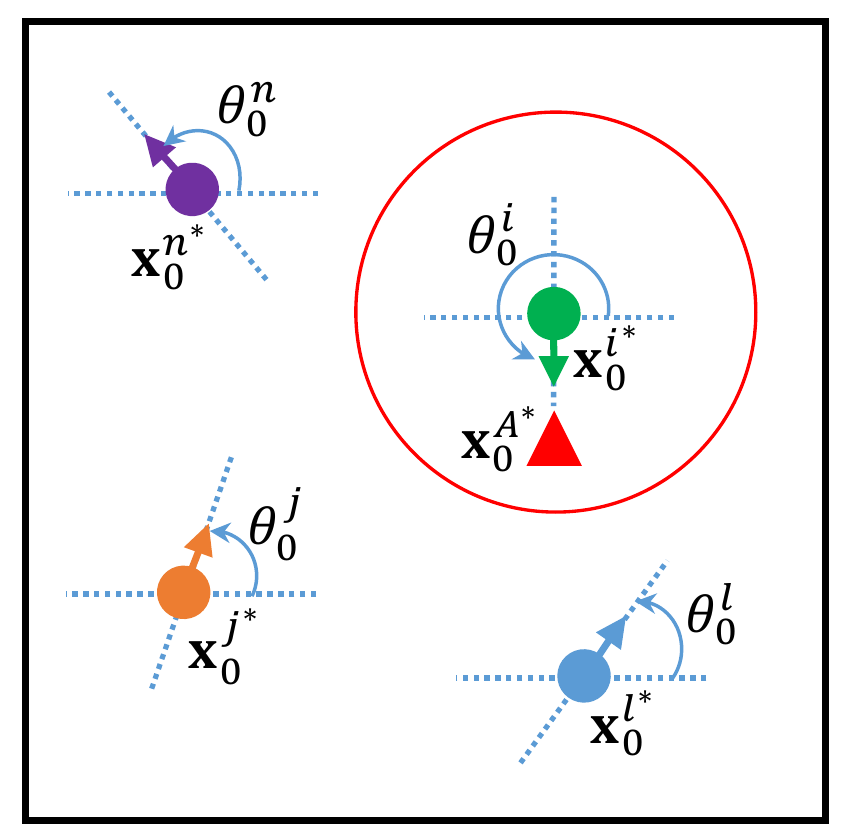}
	\caption{Initial orientations of $N=4$ agents; filled colored circles and red triangle indicate the initial locations of the agents, $i$, $j$, $l$, $n$ and the anchor,~$A$, respectively; red circle shows the communication radius of agent $i$.}
	\label{f0}
\end{figure}
Let us also consider a vision-based distance measurement process, where an agent needs to see another agent with its camera in order to find the mutual distance,
as opposed to wireless-based distance measurements that are prone to extremely large errors.
When agent $i$ moves to a new location \textcolor{black}{at time $k$, it first scans the neighborhood, i.e., rotates by  \textcolor{black}{an angle,~$0\leq\beta_k^i<2\pi$}}, in order to find (and make visual contact with) another node (agent or anchor). If the agent does not find any neighbor at time $k$, \textcolor{black}{then~$\beta_k^i=0$}. If on the other hand, agent $i$ finds another node within the communication radius, it exchanges information with that node, \textcolor{black}{then makes another rotation by an angle,~$0 \leq \alpha_k^i<2\pi$, which is randomly chosen at each iteration, and travels the distance of~$d_{k\rightarrow k+1}^i$ in the new direction}. Thus, the angle traveled by agent $i$ between time~$k$ and~$k+1$,~$\theta_{k\rightarrow k+1}^i$, can be represented~as
\begin{equation}\label{2}
\theta_{k\rightarrow k+1}^i=\alpha_{k}^i+\beta_{k}^i \qquad \textcolor{black}{k\geq0}.
\end{equation}
\textcolor{black}{Before we proceed, note that under perfect motion and noiseless distance measurements, an agent cannot find its position by direct communication with an anchor, unless it knows its angle towards the anchor with respect to a mutual frame; in contrast, the proposed framework in this paper assumes that \textit{only relative angles}  are known to each agent. For example, agent $i$ in Fig.~\ref{f0} initially finds anchor $A$ in its communication radius. However, in order to find its exact location, agent $i$ needs to know~$\theta_0^{i}$, in addition to $\mb{x}_{0}^{A\ast}$ and $\widetilde{d}_0^{iA}$. We assume that such angles are not available to the agents.}

In what follows, we first explain the procedure to track the distance between an agent \textcolor{black}{and the locations where it has exchanged information with other nodes in the past. We then show how an agent can compute the distance between a neighbor and \textcolor{black}{any such location} at the time of communication.}

\subsection{\textcolor{black}{Tracking the distance after a direct communication}}\label{gm1}
Consider an agent,~$i\in\Omega$, to fall within a distance,~$r$, of  node,~$j\in\Theta$, at some time~$k_j$; by \textcolor{black}{Assumption} {\bf A2}, agent~$i$ can measure its distance,~\textcolor{black}{$\widetilde{d}_{k_j}^{ij}$} to node~$j$, and receive $j$-th node's location estimate, hereinafter denoted as~$\mb{x}_{k_j}^j$. 
Once this information is acquired at time~$k_j$, agent~$i$ tracks the distance to the \emph{true position},~$\mb{x}_{k_j}^{j\ast}$, for all~$k\geq k_j$, even when the \textcolor{black}{two agents} move apart. In other words, agent~$i$, has the following information for each node it has communicated with:
\begin{eqnarray}\label{eq4}
\{j,k_j,\textcolor{black}{\widetilde{d}_{k}^{ij}},\mb{x}_{k_j}^j\},\qquad \textcolor{black}{k\geq k_j},j\in\Theta,
\end{eqnarray}
where~$k_j$ is the instant of the most-recent contact, \textcolor{black}{and ~$\widetilde{d}_{k}^{ij}$ is the distance between the true positions,~$\mb{x}_{k}^{i\ast}$ and~$\mb{x}_{k_j}^{j\ast},k \geq k_j$.}
This procedure is illustrated in Fig.~\ref{f2}.
\begin{figure}[!h]
	\centering
	\includegraphics[width=2.5in]{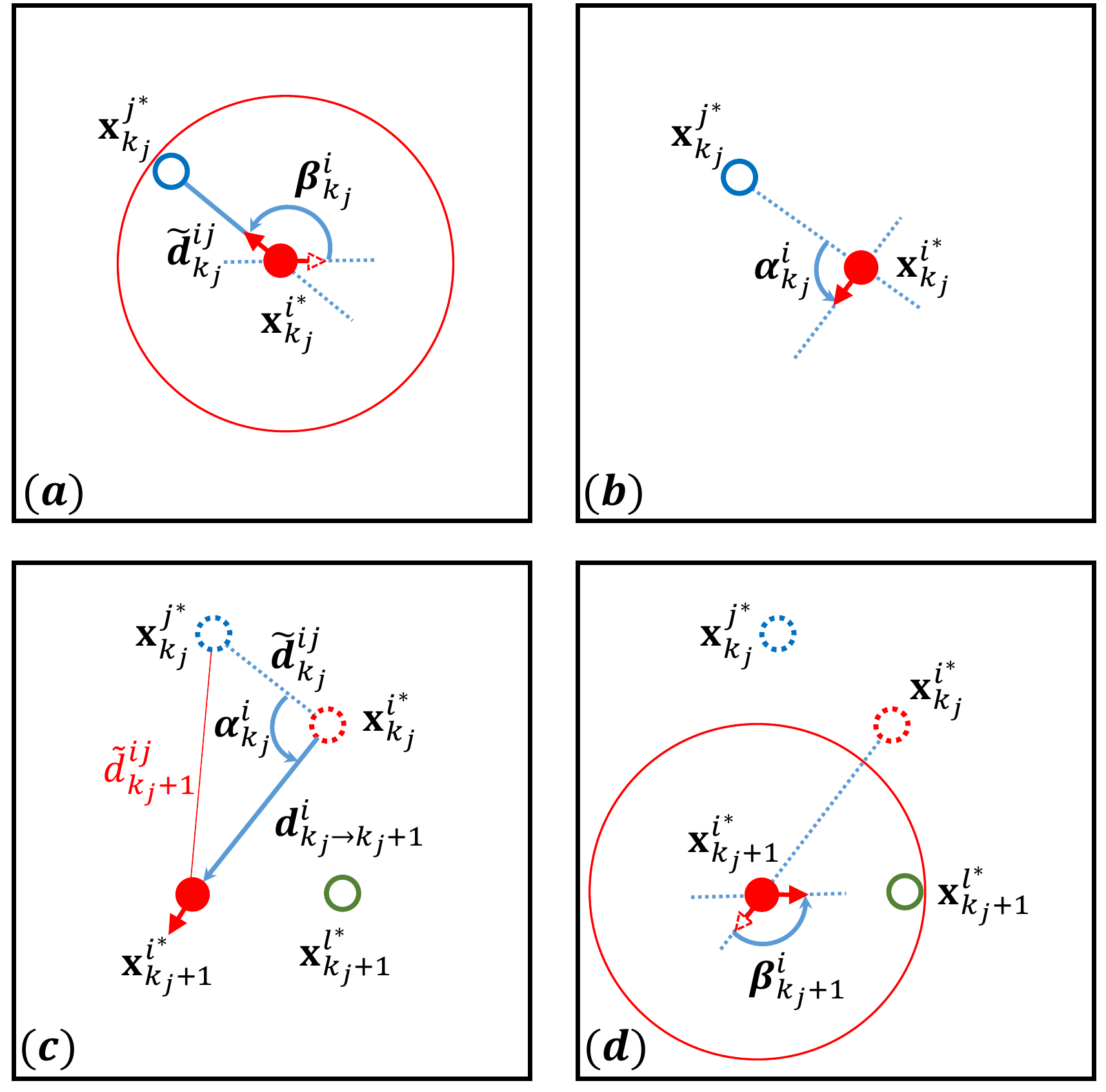}
	\caption{Distance tracking after a direct communication; agent~$i$ is indicated by red filled circle; agents $j$ and $l$ are represented by blue and green circles, respectively; red circle indicates the communication radius.}
	\label{f2}
\end{figure}
 Note that in order to make contact with agent $j$, agent $i$ has changed its orientation by $\beta_{k_j}^i$ at time $k_j$, see Fig.~\ref{f2} (a). After the two agents exchange information, agent $i$ changes its orientation by~$\alpha_{k_j}^i$, Fig.~\ref{f2} (b), and travels the distance of $d^i_{k_j\rightarrow {k_j}+1}$ in the new direction, Fig.~\ref{f2}~(c).
We now show how an agent tracks the distance to \textcolor{black}{a virtual location, where it has exchanged information with another node in the past.}

\begin{lem}\label{lem1}
	\textcolor{black}{Consider agent,~$i\in\Omega$, with true position,~$\mb{x}_{k_j}^{i\ast}$. Suppose agent $i$ communicates with node,~$j\in\Theta$, at time~$k_j$}. Let~\textcolor{black}{$\widetilde{d}_{k_j}^{ij}$} be the distance between~$i$ and~$j$ \textcolor{black}{at the time of communication}. Suppose agent~$i$ and~$j$ move apart at time~$k_j+1$, where the motion is given by~$d_{k_j\rightarrow k_j+1}^i$ and~$\theta_{k_j\rightarrow k_j+1}^i$, both known to agent~$i$.
	\noindent The distance between the true positions,~$\mb{x}_{k_j+1}^{i\ast}$ and~$\mb{x}_{k_j}^{j\ast}$, is
	{\small\begin{eqnarray}\label{3-1}
	{\widetilde{d}_{k_j+1}^{ij}} = \Big((\textcolor{black}{\widetilde{d}_{k_j}^{ij}})^{2} + ({d_{k_j\rightarrow {k_j}+1}^i})^{2}- 2{\textcolor{black}{\widetilde{d}_{k_j}^{ij}}}d_{k_j\rightarrow {k_j}+1}^i\cos(\alpha^i_{k_j})\Big)^\frac{1}{2}.
	\end{eqnarray}}
\end{lem}

\textcolor{black}{See Appendix~\ref{pr1} for the proof.}
\noindent \textcolor{black}{Note that at time ${k_j}+1$, agent~$i$ finds agent $l$ in its communication radius, hence changes its direction by $\beta_{{k_j}+1}^i$ in order to make contact with agent $l$, see Fig~\ref{f2} (d).}

\subsection{\textcolor{black}{Finding the distance with indirect communication}}\label{gm2}
We now show how agent~$i$ can use the \textcolor{black}{distance/angle} information to find the distances between a neighbor and any virtual location, where it has previously communicated with another node.
\begin{lem}\label{lem2}
	Suppose agent~$i$ has previously \textcolor{black}{made contact with} node,~$j$, hence possesses the following information:
	\begin{eqnarray*}
		\{ j,k_j,\textcolor{black}{\widetilde{d}_{k}^{ij}},\mb{x}_{k_j}^j\},\qquad& \textcolor{black}{k\geq k_j},
	\end{eqnarray*} 
	Suppose agent $i$ finds a neighbor, say agent $\ell$, at time $k_\ell > k_j$. Agent~$i$ can then find the distance between~$\mb{x}_{k_j}^{j\ast}$ and~$\mb{x}_{k_\ell}^{\ell\ast}$ as follows:
	\begin{eqnarray}\label{14}
	\widetilde{d}_{k_\ell}^{j\ell}=\sqrt{
		(\widetilde{d}_{k_\ell}^{ij})^{2}+ 
		(	\textcolor{black}{{\widetilde{d}_{k_\ell}^{i\ell}}})^{2} -
		 2\widetilde{d}_{k_\ell}^{ij}{\textcolor{black}{{\widetilde{d}_{k_\ell}^{i\ell}}}}\cos\angle(\widetilde{d}_{k_\ell}^{ij},{\textcolor{black}{{\widetilde{d}_{k_\ell}^{i\ell}}}})
		},
	\end{eqnarray}
	in which $\widetilde{d}_{k_\ell}^{ij}$ is the distance between~$\mb{x}_{k_\ell}^{i\ast}$ and~$\mb{x}_{k_j}^{j\ast}$, and~$\angle(\widetilde{d}_{k_\ell}^{ij},{\textcolor{black}{{\widetilde{d}_{k_\ell}^{i\ell}}}})$ is the angle between the two lines connecting~$\mb{x}_{k_\ell}^{i\ast}$ to~$\mb{x}_{k_j}^{j\ast}$ and~$\mb{x}_{k_\ell}^{\ell\ast}$.
\end{lem}
	\vspace{-5mm}
	\begin{figure}[!h]
		\centering
		\includegraphics[width=1.8in]{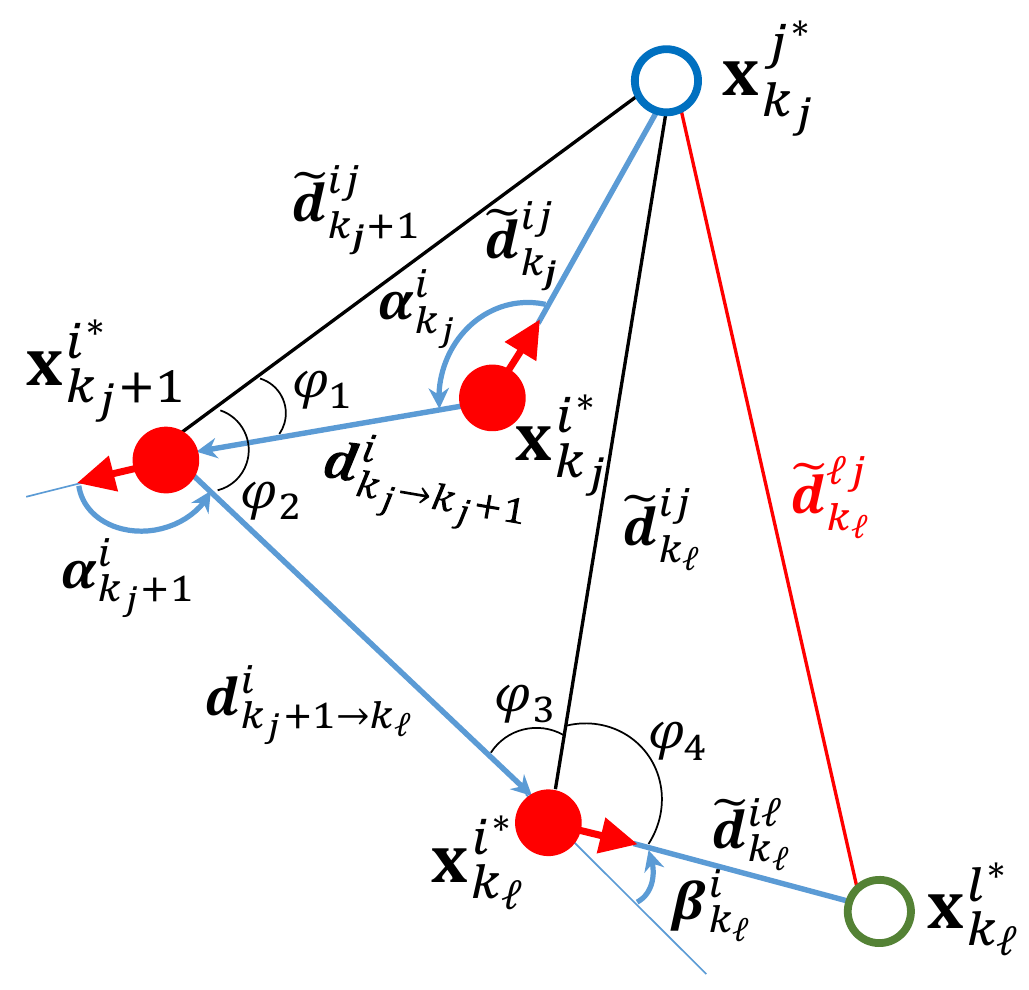}
		\caption{\textcolor{black}{At time $k_\ell$, agent $i$ finds the distance between a new neighbor, node~$\ell$, and the virtual location,~${\mb{x}}_{k_j}^{j\ast}$, where it exchanged information with node $j$ at time~$k_j <k_{\ell}$}.}
		\label{fig1-1}
	\end{figure}
\textcolor{black}{See Appendix~\ref{pr2} for the proof.}

In the next section, we use these distance tracking methods to describe the position tracking algorithm.

\section{Distributed Mobile Localization: Algorithm}\label{sec4}
Consider a network of~$N$ agents with unknown locations and~$M$ anchors, according to the motion model introduced in Section~\ref{sec2}. Let~${\mathcal{V}}_i(k)\subseteq\Theta$ be the \emph{$i$-visited set}, defined as the set of distinct nodes visited by agent,~$i\in\Omega$, up to time~$k$; and call an element in this set as \textit{$i$-visited node}. We start by introducing the notion of a virtual convex hull. 

\subsection{Virtual convex hull}\label{sec_vch}
Suppose agent~$i$ communicates with node~$j$ at time~$k_j$, and obtains the distance,~$\widetilde{d}^{ij}_{k_j}$ to~$j$, along with~$j$'s current location estimate,~$\mb{x}_{k_j}^j$, i.e.,~$j \in {\mathcal{V}}_i(k_j)$. At any time,~$k>k_j$, agents,~$i$ and~$j$, may move apart but agent~$i$ now knows\textcolor{black}{~$\widetilde{d}^{ij}_k~\forall k>k_j$}, using the geometric framework discussed in Section~\ref{gm1}. At some later time,~$k_\ell > k_j$, agent~$i$ makes contact with another node,~$\ell$, and thus obtains~$\mb{x}_{k_\ell}^\ell$ \textcolor{black}{and~$\widetilde{d}^{i\ell}_{k_\ell}$}, and keeps track of\textcolor{black}{~$\widetilde{d}^{i\ell}_k,~\forall k>k_\ell$}, thus~$j,\ell \in {\mathcal{V}}_i(k), \forall k\geq k_\ell$. Using the approach described in Section~\ref{gm2}, at time $k_\ell$, agent $i$ \textcolor{black}{also} computes $\widetilde{d}_{k_\ell}^{j\ell}$, i.e., the distance between~$\mb{x}_{k_j}^{j\ast}$ and~$\mb{x}_{k_\ell}^{\ell\ast}$.
Finally, agent~$i$ meets agent~$n$ at some~$k_n > k_\ell$, and thus now possesses the location estimates:~$\mb{x}_{k_j}^j,\mb{x}_{k_\ell}^\ell,\mb{x}_{k_n}^n$; and the distances:\textcolor{black}{~$\widetilde{d}^{iq}_k,k> k_n$}, with~$q=j,\ell,n$  (computed by \textcolor{black}{Lemma~\ref{lem1}}), along with the following distances:~$\textcolor{black}{\widetilde{d}_{k_n}^{j\ell}}$,~$\widetilde{d}_{k_n}^{nj}$, and~$\widetilde{d}_{k_n}^{n\ell}$ (computed by \textcolor{black}{Lemma~\ref{lem2}}).
At this point~$j,\ell,n \in {\mathcal{V}}_i(k), \forall k\geq k_n$, and agent~$i$ can use the $6$ distances to perform the inclusion test, described in Appendix~\ref{incl}, to check if the three visited nodes forms a virtual convex hull in which agent~$i$ lies at time~$k$.
If the test is passed, the set,~$\Theta_i(k)\triangleq\{j,\ell,n\}\subseteq\mc{V}_i(k)$, forms the virtual convex hull; otherwise, agent~$i$ continues to move \textit{and }add nodes in~$\mc{V}_i(k)$ until some combination passes the convexity test. Note that the distance between the agent and each of the virtual locations is updated every time agent~$i$ moves to a new location. However, the \textit{distance between the virtual locations} does not change unless agent~$i$ \textit{revisits} any of the previously visited nodes.

Fig.~\ref{f5-1}~(a) shows the trajectories of four agents:~$\ocircle,\Box,\pentagon,\hexagon$, over~$k=1,\ldots,9$; the time-indices are marked inside the agent symbols. 
From the perspective of agent~$\ocircle$, see Fig.~\ref{f5-1} (b): it first makes contact (communicates) with agent~$\square$, at time~$k_{\square}=2$, and then they both move apart; next, it makes contact with agents,~$\pentagon$ at~$k_{\pentagon}=4$, and~$\hexagon$ at~$k_{\hexagon}=6$. We have~$\mc{V}_{\ocircle}(2)=\{\Box\}$,~$\mc{V}_{\ocircle}(4)=\{\Box,\pentagon\}$, and~$\mc{V}_{\ocircle}(6)=\{\Box,\pentagon,\hexagon\}$, where a non-trivial convex hull becomes available at~$k=6$. However, agent~$\ocircle$ does not lie in the corresponding convex hull,~$\mc{C}(V_{\ocircle}(6))$, and cannot update its location estimate with the past neighboring estimates:~$\mb{x}^{\square}_{k_{\square}},\mb{x}^{\pentagon}_{k_{\pentagon}},\mb{x}^{\hexagon}_{k_{\hexagon}}$. At this point, agent~$\ocircle$ must wait until it \emph{either} moves inside the convex hull of~$\Box,\pentagon,\hexagon$, \emph{or} finds another agent with which the \textcolor{black}{convexity condition} is satisfied. 
Fig.~\ref{f5-1} illustrates this process in four frames.
The former is shown in Fig.~\ref{f5-1} (d), where agent~$\ocircle$ has moved inside~$\mc{C}(V_{\ocircle}(6))$ at some later time,~$k=9$; we have~$\Theta_{\ocircle}(9)=\{\Box,\pentagon,\hexagon\}$. The notion of a \emph{virtual} convex hull is evident from this discussion: an agent may only communicate with at most one agent at any given time; when the convexity condition is satisfied eventually, the updating agent may not be in communication with the corresponding nodes. Once~$\Theta_{\ocircle}(k)$ is successfully formed, agent~${\ocircle}$ updates its location \emph{linearly} using the barycentric representation in Eq.~\eqref{Eq5}, where the coefficients are computed using the distance equations in Lemmas~\ref{lem1},~\ref{lem2}, and the Cayley-Menger determinant to compute areas. After this update, agent~${\ocircle}$ removes~$\Theta_{\ocircle}(k)$ from~$\mc{V}_i({\ocircle})$, as the location estimates of the nodes in~$\Theta_{\ocircle}(k)$ have been consumed\footnote{We choose this simple strategy to remove information from~$\mc{V}_i(k)$ for convenience. Another candidate strategy is to use a \textit{forgetting factor}, which chooses the past used nodes less frequently.}. 
\begin{figure*}
	\centering
	\subfigure{\includegraphics[width=1.72in]{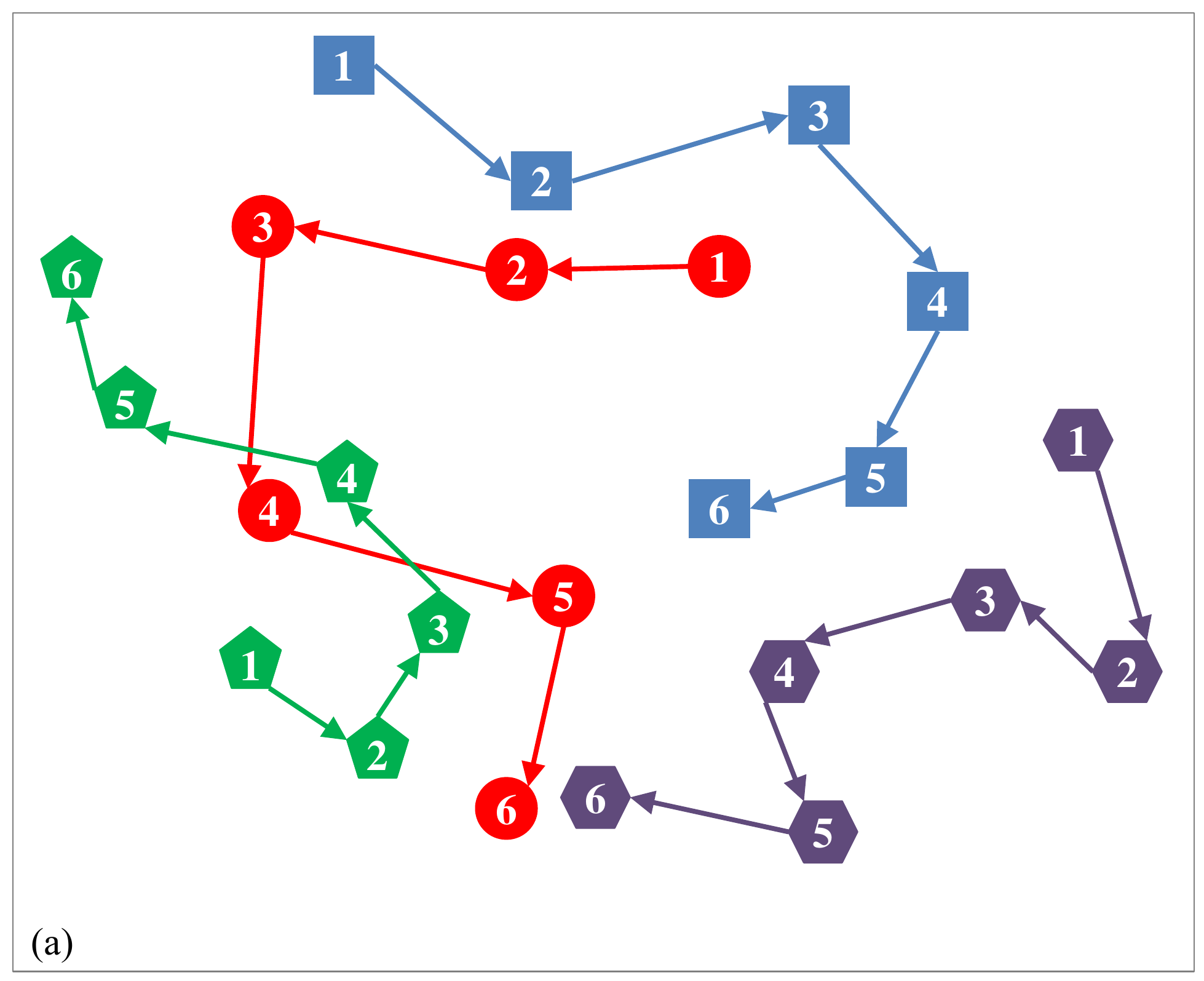}}
	\subfigure{\includegraphics[width=1.72in]{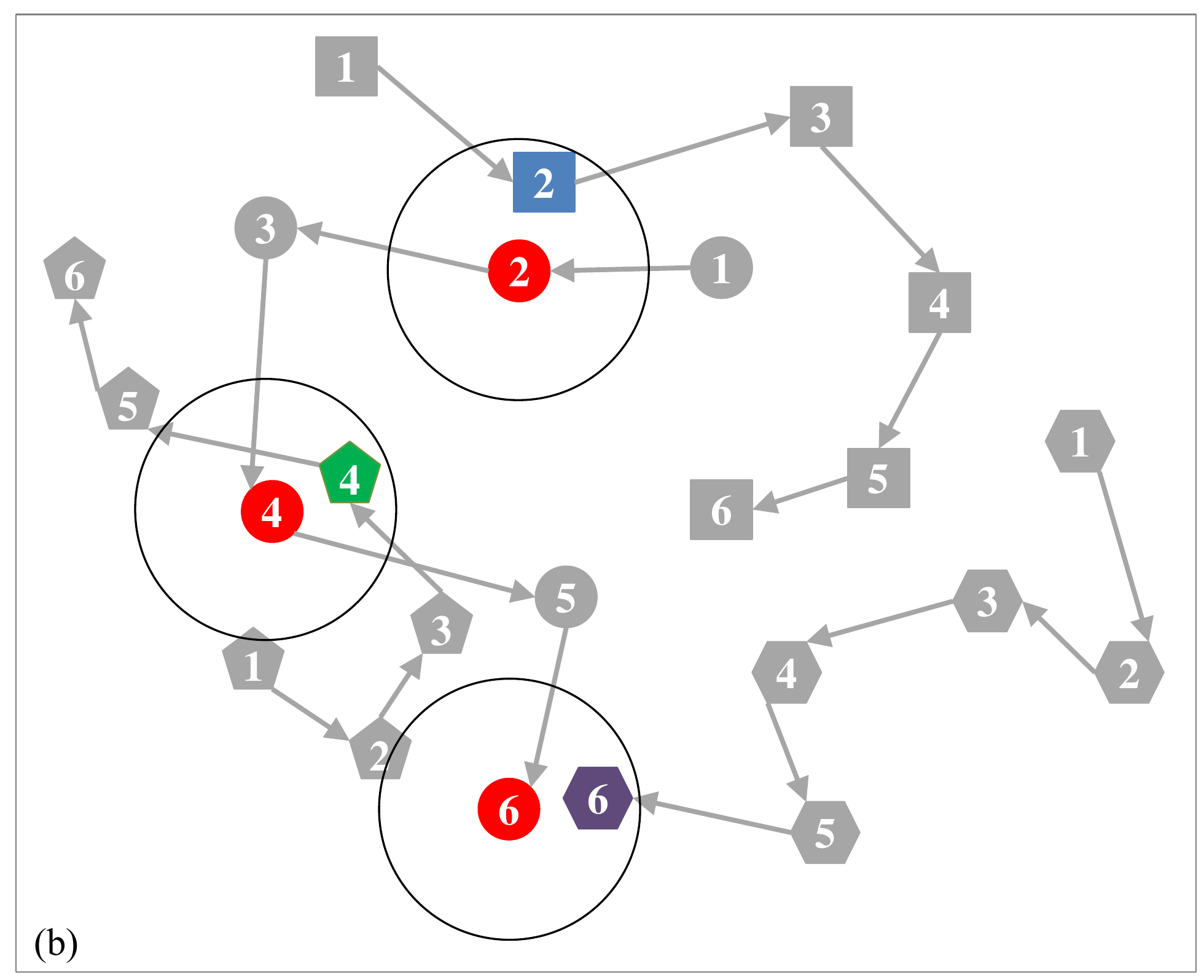}}
	\subfigure{\includegraphics[width=1.72in]{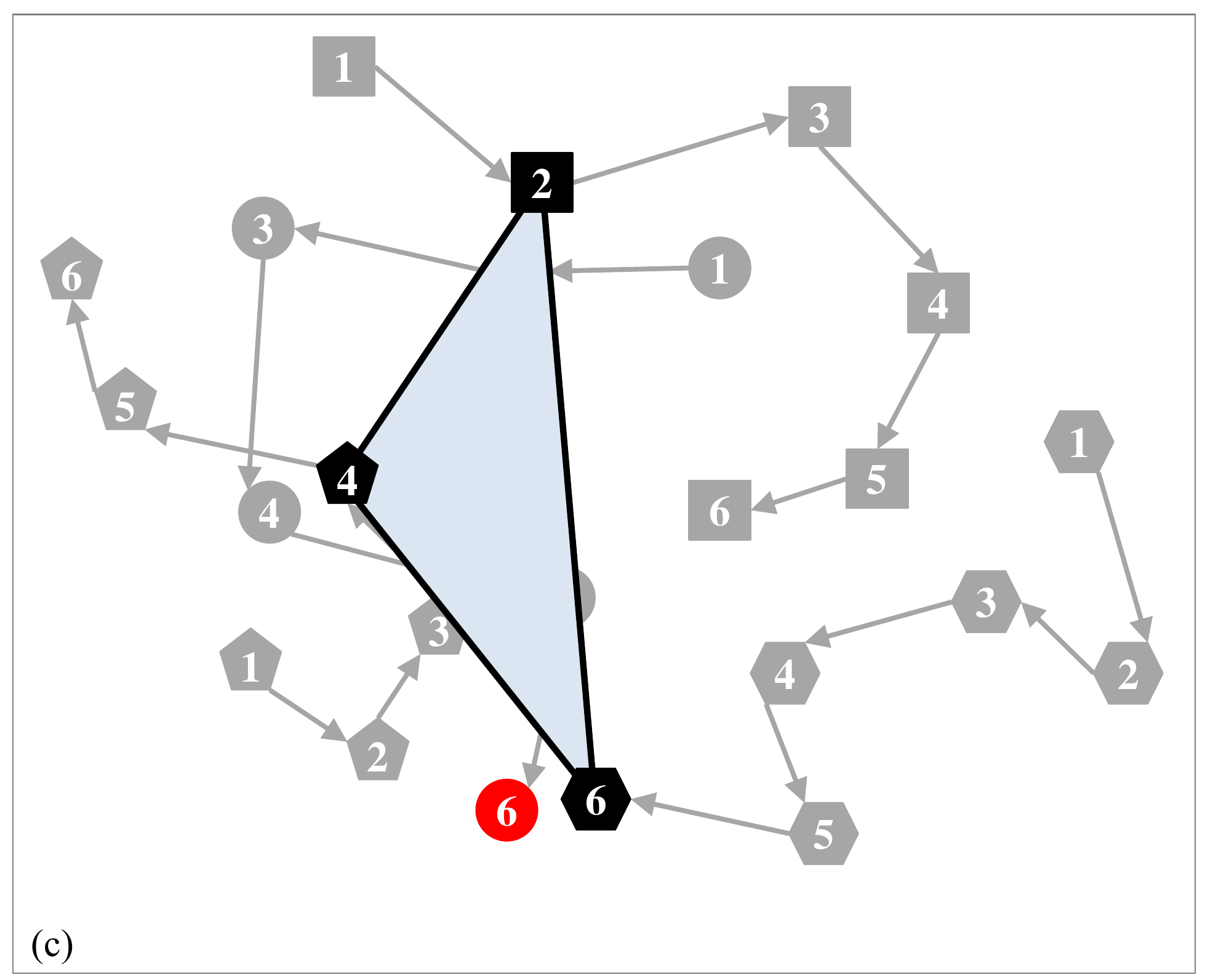}}
	\subfigure{\includegraphics[width=1.72in]{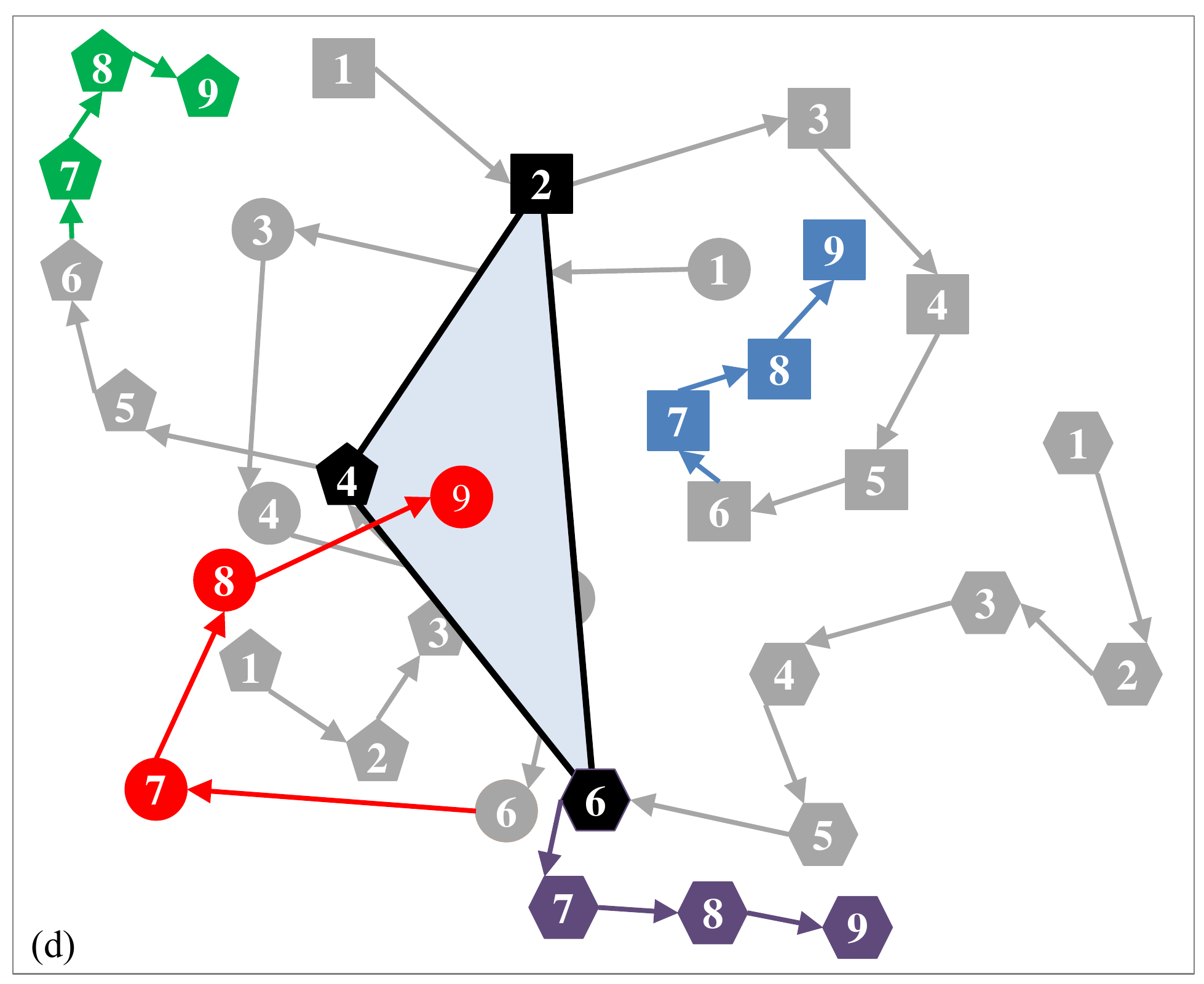}}
	\caption{Virtual convex hull with four agents:~$\ocircle,\Box,\pentagon,\hexagon$; with respect to agent~$\ocircle$: (a) Agent trajectories and time-indices;
		(b)~$\ocircle\leftrightarrow\square$ at~$k_\square=4$,~$\ocircle\leftrightarrow\pentagon$ at~$k_{\pentagon}=4$,~$\ocircle\leftrightarrow\hexagon$ at~$k_{\hexagon}=6$; circles indicate communication radius of agent~$\ocircle$; (c) Virtual convex hull of agents,~$\Box,\pentagon,\hexagon$, available at agent~$\ocircle$ at~$k=6$; (d) Trajectories at~$k>6$, test passed at~$k=9$.}
	\label{f5-1}
\end{figure*}

The following result will be useful in the sequel.
\begin{lem}\label{lem4}
For each~$i\in\Omega$, there exists a set,~$\Theta_i(k)\subseteq\mc{V}_i(k)$, such that\textcolor{black}{~$|\Theta_i(k)|=3$} and~$i\in\mc{C}(\Theta_i(k))$, for infinitely many~$k$'s.
\end{lem}

\textcolor{black}{See Appendix~\ref{pr3} for the proof.}

\subsection{Algorithm}
We now describe the localization algorithm\footnote{\textcolor{black}{Main steps of the algorithm are summarized in Appendix~\ref{psuedo}}.} in this case according to the number,~$\vert{\mathcal{V}_i(k)}\vert$, of~$i$-visited nodes in the $i$-visited set,~$\mathcal{V}_i(k)$. There are two different update scenarios for any arbitrary agent,~$i$:
\begin{enumerate}[(i)]
\item $0\leq\vert{\mathcal{V}}_i(k)\vert < 3$: Agent,~$i$, does not update its current location estimate.

\item $\vert{\mathcal{V}}_i(k)\vert \geq 3$ : Agent,~$i$, performs the inclusion test; if the test is passed the location update is applied. 
\end{enumerate}
Using the above, consider the following update:
\begin{eqnarray}\label{18}
\mb{x}^i_{k+1} = \alpha_k\mb{x}_k^i + (1-\alpha_k) \sum_{j\in\Theta_i(k)}a_k^{ij}\mb{x}_k^j + \widetilde{\mb{x}}_{k+1}^i,
\end{eqnarray}
where~$\mb{x}^i_{{k}}$ is the vector of the~$i$-th agent's coordinates at time~$k$,~$\widetilde{{\bf{x}}}^i_{{k+1}}$ is the motion vector,~$a_k^{ij}$ is the barycentric coordinate of node~$i$ with respect to the nodes~$j\in\Theta_i(k)$, and~$\alpha_k$ is such that
\begin{eqnarray}\label{alpk}
\alpha_k=\left\{
\begin{array}{ll}
1, & \forall k~|~\Theta_i(k)=\emptyset,\\
\in\left[\beta,1\right), & \forall k~|~\Theta_i(k)\neq\emptyset,
\end{array}
\right.
\end{eqnarray}
\textcolor{black}{where $\beta$ is a design parameter and~$\Theta_i(k)$ is a virtual convex hull}. \textcolor{black}{As we explain later, an updating agent receives the valuable location information only if it updates with respect to an anchor, or another agent that has previously communicated with an anchor. The non-zero self-weights assigned to the previous state of the updating agent guarantees that the agent does not completely forget the valuable information after receiving them, e.g., by performing an update where none of the agents in the triangulation set has previously received anchor information}. Note that~$\Theta_i(k)=\emptyset$ does not necessarily imply that agent~$i$ has no neighbors at time~$k$, but only that no set of neighbors meet the (virtual) convexity.
The above algorithm can be written in matrix form as
\begin{eqnarray}\label{eq1}
{\bf{x}}_{{k+1}}={\bf{P}}_{{k}}{{\bf{x}}_{{k}}}+{\bf{B}}_{{k}}{\bf{u}}_{k}+\widetilde{{\bf{x}}}_{{k+1}},\qquad k>0,
\end{eqnarray}
where~\textcolor{black}{${\bf{x}}_{{k}}$ is the vector of agent coordinates evaluated at time~$k$,~${\bf{u}}_{{k}}$ is
the vector of anchor coordinates at time $k$}, and~$\widetilde{{\bf{x}}}_{{k+1}}$ is the change in the location of agents at the beginning of the~$k-$th iteration according to the motion model. Also~${{\bf{P}}_k}$ and~${\bf{B}}_k$ are the system matrix and the input matrix of the above LTV system. 
We denote the~$(i,j)$-th element of the matrices,~${\bf{P}}_{{k}}$ and~${\bf{B}}_{{k}}$, as $({\bf{P}}_k)_{i,j}$ and~${({\bf{B}}_k)_{i,j}}$, respectively.
We can now rewrite Eq.~\eqref{18} as
\begin{eqnarray}\label{20}
\mb{x}^i_{k+1} &=& \alpha_k\mb{x}_k^i + (1-\alpha_k) \left(\sum_{j\in\Theta_i(k)\cap\Omega}\textcolor{black}{{a}_k^{ij}}\mb{x}_k^j\right),\nonumber\\ &+&(1-\alpha_k) \left(\sum_{m\in\Theta_i(k)\cap\kappa}\textcolor{black}{{a}_k^{im}}\mb{u}_k^m\right)+\widetilde{\mb{x}}_{k+1}^i,\qquad\label{24}
\end{eqnarray}
where \textcolor{black}{$\mb{u}_k^m$ denotes the known coordinates of the $m$-th anchor at time $k$}. \textcolor{black}{Thus,~$({\bf{P}}_k)_{i,i}=\alpha_k$ as defined in Eq.~\eqref{alpk}, and
\begin{eqnarray}
({\bf{P}}_k)_{i,j}&=& (1-\alpha_k)a_k^{ij} ~~\mbox{if}~~j \in \Theta_i(k)\cap\Omega,\nonumber\\
({\bf{B}}_k)_{i,m}&=& (1-\alpha_k)a_k^{im} ~~\mbox{if}~~m \in \Theta_i(k)\cap\kappa,
\end{eqnarray}}\textcolor{black}{If the triangulation set for agent $i$ at time $k$ does not contain the $j$-th agent or the $m$-th anchor, the corresponding elements in system and input matrices,~$({\bf{P}}_k)_{i,j}$ and~$({\bf{B}}_k)_{i,m}$ are zero.}
Note that Eq.~\eqref{alpk} immediately implies that the self-weight at each agent is always lower bounded, i.e.,
\begin{eqnarray}\label{B0eq}
0<\beta\leq ({\bf{P}}_k)_{i,i}\leq1, \forall k, i. 
\end{eqnarray} 
Since accurate information is only injected via the anchors, it is reasonable (and necessary) to set a lower bound on the weights assigned to the anchor states. In particular, we make the following assumption. 

{\bf A3: Anchor contribution}. For any update that involves an anchor, i.e., for any~$({\bf{B}}_k)_{i,m} \neq 0$, we assume that
\begin{eqnarray}\label{23}
0 < \alpha \leq ({\bf{B}}_k)_{i,m},\qquad \forall k, i\in\Omega, {m\in\Theta_i(k)\cap\kappa},
\end{eqnarray}
where $\alpha$ is the minimum anchor contribution, \textcolor{black}{see Section~\ref{ca}.}}}
Assumption~{\bf{A3}} implies that if there is an anchor in the (virtual) convex hull, \textcolor{black}{it always contributes a certain amount of information}. In other words, the updating agent has to lie in an \textit{appropriate} position inside the convex hull. Assumption {\bf{A3}} states that the weight assigned to the anchor (which comes from barycentric coordinates) should be at least~$\alpha$. Therefore, the area (in $\mathbb{R}^{2}$) of the triangle corresponding to the anchor must take an adequate portion of the area of the \textcolor{black}{whole convex hull triangle}. This is illustrated in~Fig.~\ref{figg}, where the updating agent lies inside a virtual convex hull, consisting of an anchor (node $j$), and two other \textcolor{black}{agents,~$m$ and $\ell$}. Node $i$ \textcolor{black}{has} communicated with the anchor and the two other agents, at time instants,~$k_j$,~$k_m$, and~$k_\ell$, respectively. 

\begin{figure}[!h]
	\centering
	\includegraphics[height=1.25in]{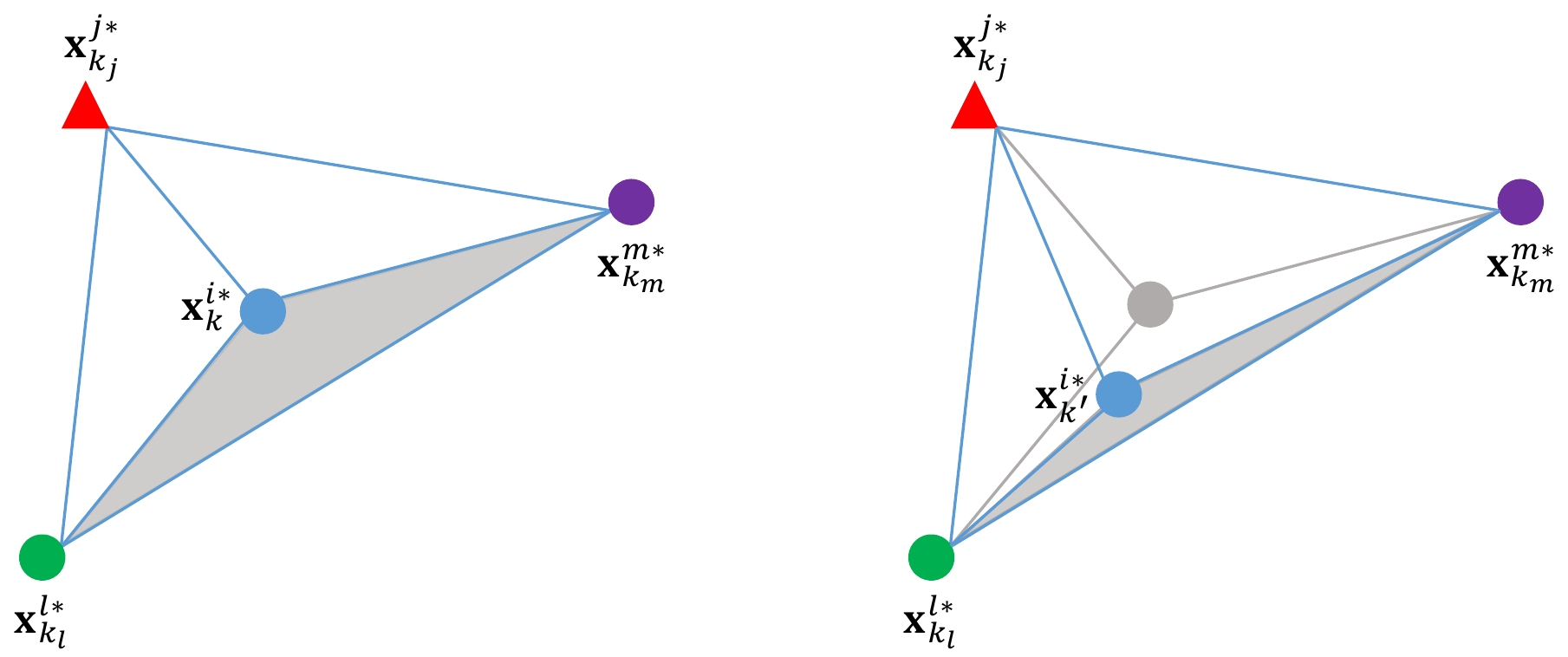}
	\caption{(Left) \textcolor{black}{At time $k$}, agent $i$ is located on the threshold boundaries, which assigns the minimum weight, $\alpha$, to the anchor. (Right) \textcolor{black}{At time $k^{\prime}$} agent $i$ is located in an inappropriate location inside the convex hull.}
	\label{figg}
\end{figure}
\noindent Fig.~\ref{figg} (Left) illustrates the position of agent $i$ at time $k$,~${\bf{x}}^{i}_{k}$ \textcolor{black}{on the} threshold boundaries, such that
\begin{eqnarray}\label{37}
({\bf{B}}_k)_{i,j} = \frac{A_{\Theta_i(k)\cup\{i\}\setminus j}}{A_{\Theta_i(k)}}=\alpha.
\end{eqnarray}
Fig.~\ref{figg} (Right) on the other hand shows that if the agent lies in \textcolor{black}{${\bf{x}}^{i\ast}_{k^\prime}$} (or any other position within the triangle of ${\bf{x}}^{i\ast}_{k}$, ${\bf{x}}^{l\ast}_{k_l}$, and ${\bf{x}}^{m\ast}_{k_m}$), the left hand side of Eq.~\eqref{37} becomes less than $\alpha$, and Eq.~\eqref{23} does not hold. Since the corresponding update does not provide enough valuable information for agent $i$, no update occurs in this case.

With the lower bounds on both the self-weights, Eq.~\eqref{B0eq}, and the anchor weights, Eq.~\eqref{23}, we  note that at time~$k$, the matrix of barycentric coordinates with respect to agents with unknown locations, i.e., the system matrix~${\bf{P}}_{k}$, is either
\begin{enumerate}[(i)]
\item \emph{identity}, when no update occurs; or,
\item \emph{identity except a stochastic~$i$-th row}, when there is no anchor in the virtual convex hull, i.e.,~$\Theta_i(k)\cap\kappa=\emptyset$; or, 
\item \emph{identity except a strictly sub-stochastic~$i$-th row}, when there is at least one anchor in the virtual convex hull\textcolor{black}{\footnote{\textcolor{black}{Note that if an agent lies inside a virtual convex hull of three anchors, then practically by assigning a zero weight on its past, the agent can find its exact location and take the role of an anchor, which could subsequently increase the convergence rate of the algorithm.}}}.
\end{enumerate}

In the next section, we provide sufficient conditions under which the iterative localization algorithm, Eq.~\eqref{eq1}, tracks the true agent locations. Before we proceed, let us make the following definitions to clarify what we mean by (strictly sub-) stochasticity throughout this~paper:
\begin{definition}
	A non-negative, stochastic matrix is such that all of its row sums are one. A non-negative, strictly sub-stochastic matrix is such that it has at least one row that sums to strictly less than one and every other row sums to at most one.
\end{definition}

\section{Distributed Mobile Localization: Analysis}\label{sec5}
In this section, we address the challenge on the analysis of LTV systems with potentially non-deterministic system matrices. We borrow the following result on the asymptotic stability of LTV systems from~\cite{DBLP:journals/corr/SafaviK14}.

\subsection{Asymptotic stability of LTV systems}\label{ltv}
Consider an LTV system:~$\mb{x}_{k+1}=\mb{P}_k\mb{x}_k$ such that the system matrix,~$\mb{P}_k$, is time-varying and non-deterministic. The system matrix,~$\mb{P}_k$, represents at most one state update, say the~$i$-th state, for any~$k$, i.e., at most one row,~$i$, of~$\mb{P}_k$ is different from identity and can be either stochastic or strictly sub-stochastic, not necessarily in any order and with arbitrary elements as long as the bounds in \textcolor{black}{\textbf{(A3)}, \textbf{(B0), and (B1)}} in the following are satisfied: 

{\bf B0}: If the updating row,~$i$, in~$\mb{P}_k$ sums to~$1$, then
\begin{eqnarray}\label{bnd1}
0 < \beta_1 \leq (\mb{P}_{k})_{i,i},\qquad\beta_1\in\mbb{R},
\end{eqnarray}

{\bf B1}: If the updating row,~$i$, in~$\mb{P}_k$ does not sum to~$1$, then
\begin{eqnarray}\label{bnd2}
\sum_j({\bf{P}}_{k})_{i,j} \leq\beta_2 < 1,\qquad\beta_2\in\mbb{R}.
\end{eqnarray}

To analyze the asymptotic behavior of an LTV system with such system matrices, we \textcolor{black}{utilize} the notion of a slice,~${\bf{M}}_j$, which is the smallest product of consecutive system matrices,
such that the entire chain of systems matrices is covered by non-overlapping slices,~i.e., $\prod_t {\bf{M}}_t = \prod_{k}{\bf{P}}_k,$ \emph{and} each slice has a subunit infinity norm (maximum row sum), i.e.,~$\|{\bf{M}}_t\|_\infty<1,\forall t$.

\begin{figure}[!h]
	\centering
	\includegraphics[height=1.75in]{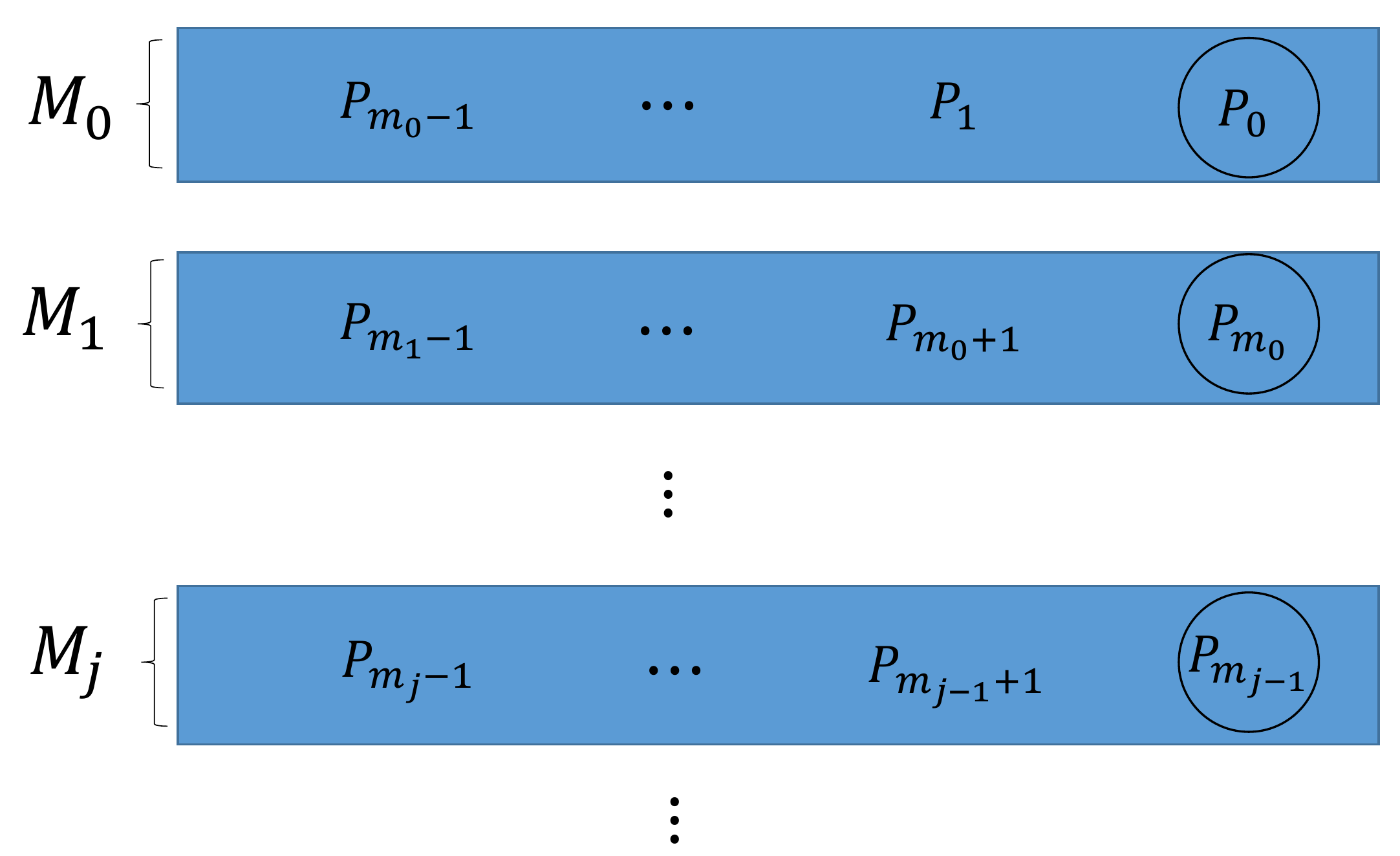}
	\caption{Slice representation.}
	\label{figs}
\end{figure}
Slices are initiated by strictly sub-stochastic system matrices, and terminated after all row sums are strictly less than one. \textcolor{black}{Slice representation is depicted in Fig.~\ref{figs}, where the rightmost system matrices (encircled in Fig.~\ref{figs}) of each slice, i.e.,~$\mb{P}_0,~\mb{P}_{m_0},~\ldots,~\mb{P}_{m_{j-1}} \ldots$, are strictly sub-stochastic. The length of a slice is defined as the number of matrices forming the slice, and for the~$j$-th slice length we have $\vert {\bf{M}}_{j} \vert = m_j - m_{j-1}, m_{-1}=0.$} Ref.~\cite{DBLP:journals/corr/SafaviK14} also shows that the upper bound on the infinity norm of a slice is further related to the length of the slice, i.e., the number of matrices forming the slice. The following theorem characterizes the asymptotic stability of the above LTV system.

\begin{thm}\label{thm0}
With Assumptions,~{\bf{B0-B1}}, the system,~$\mb{x}_{k+1}=\mb{P}_k\mb{x}_k$, converges to zero, i.e.,~$\lim_{k\ra\infty}\mb{x}_k = \mb{0}_N$,
%
%
if for every~$i \in \mathbb{N}$, there exists \textcolor{black}{a subset,~$J$,} of slices such that
\begin{equation}\label{growth}
\exists {\bf{M}}_j \in \textcolor{black}{J}:~~\vert {\bf{M}}_j \vert \leq \frac{1}{\ln\left({\beta_1}\right)}\ln\left(\frac{1 - e^{(-\gamma_2i^{-\gamma_1})}}{1-\beta_2}\right)+1,
\end{equation}
 for some $\gamma_1 \in [0,1],~\gamma_2 >0$. For any other slice,~\textcolor{black}{$j\notin J$} we have $|{\bf{M}}_j|<\infty$.
\end{thm}
The proof is available in our prior work,~\cite{DBLP:journals/corr/SafaviK14}. Here, we explain the intuition behind the above theorem.
%
%
%
\textcolor{black}{The asymptotic stability of the system is guaranteed if all (or and infinite subset of) slices have bounded lengths. Theorem~\ref{thm0} states that the system is asymptotically stable even in the non-trivial case,} where there exist an infinite subset of slices whose lengths are not bounded, but do not grow faster than the exponential growth in Eq.~\eqref{growth}. As detailed in~\cite{DBLP:journals/corr/SafaviK14}, if the lengths of an infinite subset of slices follow Eq.~\eqref{growth}, the infinite product of slices goes to a zero matrix, and the system is asymptotically stable. Note that in the RHS of Eq.~\eqref{growth}, the first~$\ln$ is negative; for the bound to remain meaningful, the second~$\ln$ must also be negative that requires~$\beta_2<e^{(-\gamma_2i^{-\gamma_1})}$, which holds for any value of $\beta_2\in[0,1)$ by choosing an appropriate $\gamma_2>0$. \textcolor{black}{Please see~\cite{DBLP:journals/corr/SafaviK14} for a detailed discussion on Eq.~\eqref{growth} and its parameters.} With the help of this theorem, we now analyze the convergence of Eq.~\eqref{eq1} in the following section.

\subsection{Convergence Analysis}\label{ca}
\textcolor{black}{We now adapt the above LTV results to the distributed localization setup described in Section~\ref{sec2}. \textcolor{black}{Recall that the $i$-th row of the system matrix,~$\mb{P}_k$ in Section~\ref{ltv}, collects the agent-to-agent barycentric coordinates} corresponding to agent~$i$ and its neighbors.
	We now relate the slice representation in Theorem~\ref{thm0} to the information flow in the network: \textcolor{black}{Each slice,~${\bf{M}}_j$,} is initiated with a strictly sub-stochastic update, i.e., when one agent with unknown location~\textit{directly} receives information from an anchor by having this anchor in its virtual convex hull. On the other hand, \textcolor{black}{a slice,~${\bf{M}}_j$,} is terminated after the information from the anchor(s) is propagated through the network and reaches every agent either directly or indirectly. Here, \emph{directly} means that an agent has an anchor in its virtual convex hull; while \emph{indirectly} means that an agent has a neighbor in its \textcolor{black}{virtual} convex hull, which has previously received the information  \textcolor{black}{(either directly or indirectly) from an anchor}.
	Once the anchor information reaches every agent in the network, the slice notion and Theorem~\ref{thm0} provide the conditions on the rate, Eq.~\eqref{growth}, at which this information should propagate \textcolor{black}{for convergence.}} \textcolor{black}{We proceed with the following lemma.}
	
\begin{lem}\label{lem5}
Under the conditions {\bf{A0-A3}} and no noise, the product of system matrices,~$\mb{P}_k$'s, in the LTV system, Eq.~\eqref{eq1}, converges to zero if the condition in Theorem~\ref{thm0} holds.
\end{lem}
\begin{proof}
First, we need to \textcolor{black}{show that the Assumptions~{\bf B0-B1} follow from~{\bf A0-A3}}. First, note that~{\bf B0} is immediately verified by Eq.~\eqref{B0eq}, assuming $\beta_1=\beta$. \textcolor{black}{Next, we note that~{\bf B1}} is implied by~{\bf A3}. This is because if~$\Theta_i\cap\kappa$ is not empty, we can write
\begin{eqnarray}\label{30}
\sum_{j\in\Theta_i(k)\cap\Omega}{{({\bf{P}}_{k})}_{i,j}} = 1-\sum_{m\in\Theta_i(k)\cap\kappa}{{({\bf{B}}_{k})}_{i,m}},
\end{eqnarray}
where we used the fact that \textcolor{black}{the} barycentric coordinates sum to one.
When there is only one anchor among the \textcolor{black}{nodes forming the virtual convex hull}, and the minimum weight is assigned to this anchor, Eq.~\eqref{23}, the right hand side of Eq.~\eqref{30} is maximized. This provides an upper bound on the~$i$-th row sum:
\begin{eqnarray}\label{31}
\sum_{j\in\Theta_i(k)\cap\Omega}{{({\bf{P}}_{k})}_{i,j}} \leq 1-\alpha <1,
\end{eqnarray}
which ensures~{\bf B1} with~$\beta_2=1-\alpha$. In addition, Lemma~\ref{lem4} ensures that each agent updates infinitely often with different neighbors. Subsequently, each slice is completed after all agents receive anchor information (at least once) either directly or indirectly, and the asymptotic convergence of Eq.~\eqref{eq1} follows under the conditions in Theorem~\ref{thm0}.
\end{proof}
\noindent Note that Assumption {\bf{B0}}, which is equivalent to Eq.~\eqref{B0eq}, implies that each agent remembers its past information. If a lower bound on the self-weights is not assigned, an agent may lose valuable information when it updates with other agents that have not previously \textcolor{black}{updated} 
(directly or indirectly) 
with an anchor. On the other hand, Assumption {\bf{B1}} restricts the amount of unreliable
information added in the network by the agents, when an anchor is involved in an update.

The following theorem completes the localization algorithm for mobile multi-agent networks. 

\begin{thm}\label{th2}
\textcolor{black}{Under the Assumptions {\bf A0-A3} and no noise, for any
(random or deterministic) motion that satisfies Eq.~\eqref{growth}, Eq.~\eqref{eq1} asymptotically tracks the exact agent locations.}
\end{thm}
\begin{proof}
In order to show the convergence to the true locations, we show that the error between the location estimate,~$\mb{x}_k$, and the true location,~$\mb{x}^*_k$, goes to zero. To find the error dynamics, note that the true agent locations follow:
\begin{eqnarray}\label{eq2}
{\bf{x}}^{*}_{{k+1}}={\bf{P}}_{{k}}{{\bf{x}}^{*}_{{k}}}+{\bf{B}}_{{k}}{\bf{u}}_k+\widetilde{{\bf{x}}}_{{k+1}}.
\end{eqnarray}
Subtracting Eq.~\eqref{eq1} from Eq.~\eqref{eq2}, we get the network error
\begin{eqnarray}\label{eq3}
\mb{e}_{k+1}\triangleq{\bf{x}}^{*}_{{k+1}}-{\bf{x}}_{{k+1}}={\bf{P}}_{{k}} ({\bf{x}}^{*}_{{k}}-{\bf{x}}_{{k}})={\bf{P}}_{{k}}\mb{e}_{{k}},
\end{eqnarray}
which goes to zero when
\begin{eqnarray}\label{eq7}
\lim_{k \rightarrow \infty} \prod_{l=0}^{k} {\bf{P}}_{l}=\mb{0}_{N\times N},
\end{eqnarray}
\textcolor{black}{from Lemma~\ref{lem5}.}
\end{proof}

\textcolor{black}{Before we proceed, it is reasonable to comment on the choice of parameter $\alpha$; The proposed localization algorithm is proved, both in theory and simulations, to converge for any value of $0<\alpha<1$. However, the convergence rate of the algorithm is affected by the choice of $\alpha$ as follows. Choosing~$\alpha$ arbitrarily close to zero corresponds to an infinitely large upper-bound on the length of a slice (this can be verified by replacing~$\alpha=1-\beta_2$ by zero on the \textcolor{black}{right hand side of} Eq.~\eqref{growth}. On the other hand, by setting $\alpha$ arbitrarily close to one, an agent has to get arbitrarily close to an anchor in order to perform an update with respect to the anchor (see Fig.~\ref{figg}), which again leads to arbitrary large number of iterations required for each slice to complete. In summary, there is a trade-off in the choice of $\alpha$, between receiving more information from an anchor at the time of an update, which is the case for $\alpha$ values closer to $0$, and increasing the chance of an update with an anchor, which requires for $\alpha$ to be closer to $1$. A proper choice can be made by considering the motion model,  the communication protocol, and the number of available anchors in the network.}

The following theorem characterizes the number of anchors.
\begin{thm}\label{thm3}
Under Assumptions~{\bf A0-A3} and no noise, Eq.~\eqref{eq1} tracks the true agent locations in~$\mbb{R}^2$, when \textcolor{black}{the number of agents, $N$, and the number of anchors, $M$, follow:}
\begin{eqnarray}
M&\geq&1,\label{32}\\
N+M&\geq& 4.
\end{eqnarray}
\end{thm}
\begin{proof}
\textcolor{black}{Let us} first consider the requirement of at least one anchor. Without an anchor, a strictly sub-stochastic row never \textcolor{black}{appears in~$\mb{P}_k$, making ~$\mb{P}_k$ stochastic at each time and hence the infinite product of~$\mb{P}_k$'s is also stochastic and not zero. With at least one anchor,} strictly sub-stochastic rows, following Eq.~\eqref{31}, appear in~$\mb{P}_k$'s, and \textcolor{black}{zero convergence of the error dynamics, Eq.~\eqref{eq3}, follows.} Next, exactly~$3$ nodes (agents and/or anchors) are required to form a (virtual) convex hull in\textcolor{black}{~$\mathbb{R}^{2}$}. Thus, when the total number of nodes,~$N+M$, is $3$ or less, no agent can find~$3$ other nodes to find (virtual) convex hulls. \textcolor{black}{Using the same argument as the one in the proof of Lemma~\ref{lem4}, we can show that with at least one anchor and at least~$4$ total nodes, any agent with unknown location infinitely finds itself in arbitrary (virtual) convex hulls.} Thus, Theorem~\ref{thm0} is applicable and the proof follows. 
\end{proof}

\section{Mobile localization under imperfect measurements}\label{noise}
The noise on distance measurements and motion degrades the performance of the localization algorithm, as expected, and in certain cases the location error is as large as the region of motion; this is shown experimentally in Section~\ref{SEC7-B}. In what follows, we discuss the modifications, {\bf{M1-M3}}, to the proposed algorithm to address the noise in case of motion,~${\widehat{\mb{x}}_{k}^i}$, and on the distance measurements,~$\widehat{d}^{ij}_k$. 
\subsection{Agent contribution}
If an agent is located close to the boundaries of a (virtual) convex hull, noise on distance measurements may lead to incorrect inclusion test, {\textcolor{black}{Appendix~\ref{incl}}, results, i.e., a set of visited neighbors may fail the convexity test while the agent is indeed located inside the convex hull, or vice versa. To avoid such scenarios, we generalize Assumption {\bf{A3}} to non-anchor agents. In particular, we restrict an agent from performing an update, unless it is located in a proper position inside a candidate (virtual) convex hull. We augment Assumption {\bf{A3}} by the following:

{\bf{M1:}} For any agent, $j \in \Omega$, involved in an update, i.e., for any~$({\bf{P}}_k)_{i,j} \neq 0$, we assume that
\begin{eqnarray}\label{ac}
0 < {\alpha}^{\prime} \leq ({\bf{P}}_k)_{i,j},\qquad \forall \textcolor{black}{k,i}\in\Omega, {j\in\Theta_i(k)\cap\Omega},
\end{eqnarray}
where $i$ is the updating agent's index, and ${\alpha}^{\prime}$ is the minimum agent contribution. In other words, for agent $i$ to perform an update, the area of the corresponding inner triangle 
has to constitute a minimum proportion, ${\alpha}^{\prime}$, of the area of the convex hull triangle. This is shown in Fig.~\ref{fign} (Left); 
\begin{figure}[!h]
	\centering
	\includegraphics[height=1in]{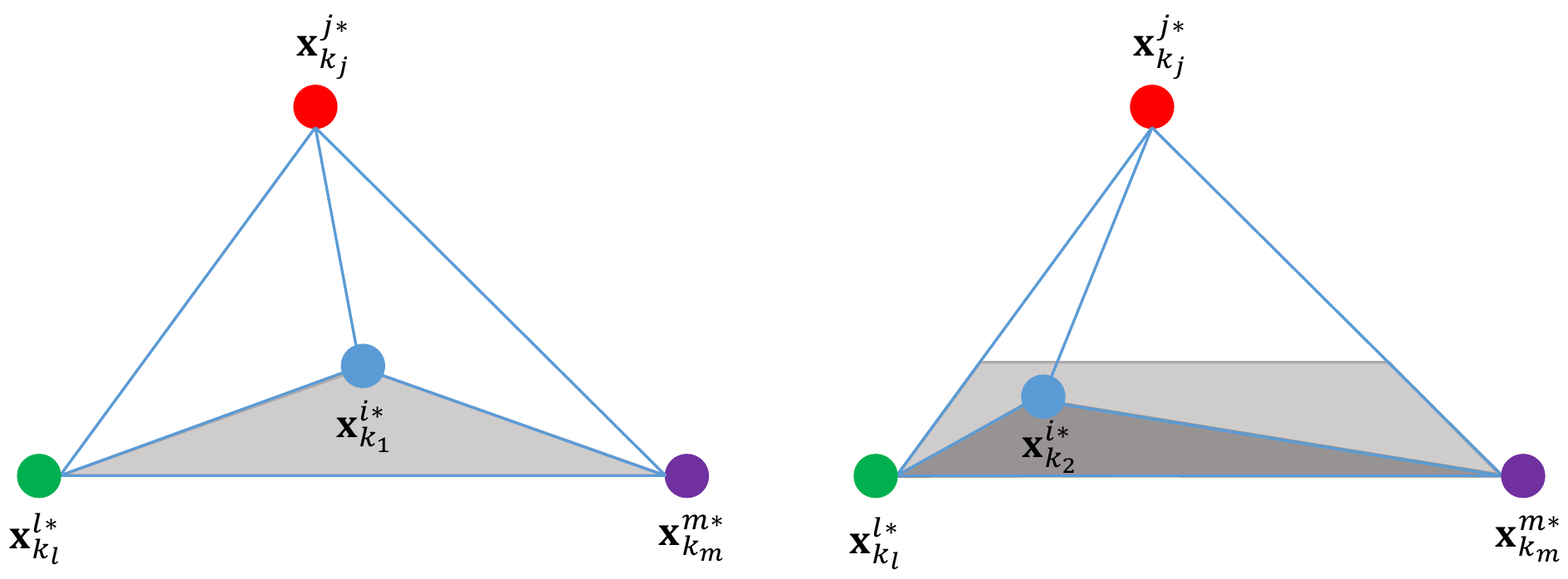}
	\caption{(Left) Agent $i$ is located on the threshold boundaries, which assigns the minimum weight, $\alpha$, to the anchor. (Right) Agent $i$ is located in an inappropriate location inside the convex hull.}
	\label{fign}
\end{figure}
At time $k_1$, agent~$i$ is located in the virtual convex hull of the nodes $j$, $m$ and $\ell$, such that the ratio of the area of the 
shaded triangle to the area of the convex hull triangle
is $\alpha^\prime$. 
Agent $i$ can perform an update at time $k_1$. The upper side of the shaded trapezoid in Fig.~\ref{fign} (Right) is the threshold boundary. If agent $i$ remains in the same virtual convex hull, but moves inside the shaded trapezoid, e.g., to ${\bf{x}}^{i\ast}_{{k_2}}$ at time $k_2$, Eq.~\eqref{ac} will not hold and no update will occur \textcolor{black}{at} that time.
Clearly, Assumption~{\bf{M1}} can be relaxed if the noise is negligible. Otherwise, the value of~${\alpha}^{\prime}$ can be adjusted according to the amount of noise on distance measurements.

\subsection{Inclusion test error}
Even if \textcolor{black}{Assumption {\bf{M1}}} is satisfied, the inclusion test results may not be accurate due to the noise on the motion, which corresponds to imperfect location updates \textit{at each and every iteration}. To tackle this issue, we propose the following modification to the algorithm:

{\bf{M2:}} If the inclusion test is passed at time $k$ by a triangulation set, $\Theta_i(k)$,  agent $i$ performs an update only if 
\begin{eqnarray}\label{iterror}
\textcolor{black}{\epsilon_k^{i}=\vert \frac{\sum_{j\in\Theta_i(k)}A_{\Theta_i(k)\cup\{i\}\setminus j} - A_{\Theta_i(k)}}{A_{\Theta_i(k)}}\vert \leq\epsilon,}
\end{eqnarray}
in which~\textcolor{black}{$\epsilon_k^{i} \geq 0$} is the \textit{inclusion test relative error} for agent $i$ at time $k$, and \textcolor{black}{$\epsilon \geq 0$} is a design parameter.

\subsection{Convexity}
Finally, recall that if the inclusion test is passed at time~$k$, by a triangulation set, $\Theta_i(k)$, the updating agent, $i$, updates its location estimate,~$\mathbf{x}_k^{i\ast}$, as a convex combination of the location estimates of the nodes in $\Theta_i(k)$. In order to guarantee the convexity in the updates \textcolor{black}{in presence of noise}, we consider the following modification to the algorithm:

{\bf{M3:}} If the inclusion test is passed at time $k$ by a triangulation set, $\Theta_i(k)$,  and Eq.~\eqref{iterror} holds at time $k$, agent $i$ (randomly) chooses one of the $i$-visited nodes from the triangulation set~$\Theta_i(k)$, say agent $j$, and finds the corresponding weight,~$a_{k}^{{ij}}$, as follows: assuming that $a_k^{im} + a^{in}_k < 1$,
\begin{equation}
a_{k}^{{i}j}= 1-a_{k}^{{i}m}-a_{k}^{{i}n},\qquad \{j,m,n\}\in{\Theta_i(k)}.
\end{equation}

\section{Simulations}\label{sec7}
In this section, we first provide simulation in noiseless scenarios and then examine the effects of noise.

\subsection{Localization without noise}
We now consider the noiseless scenarios, i.e., we assume that~$n_k^i=0$ and $r_k^{ij}=0$, in Eqs.~\eqref{Eq2} and~\eqref{Eq3}. In the beginning, all nodes (agents and anchor(s)) are randomly deployed within the region of $x\in[-5~15],~y\in[-5~15]$ in $\mathbb{R}^{2}$. \textcolor{black}{For the simulations, we consider \textit{random waypoint mobility} model,\cite{camp2002survey}, for the motion in the network, which has been the predominant motion model in the localization literature over the past decade, see e.g.,~\cite{rudafshani2007localization,hu2004localization,dil2006range,wang2009sequential,7496598}. We note that better performance may be achieved if the agents follow some deterministic (mission-related) motion model. However, our localization algorithm converges if the motion is such that the slice lengths grow slower than the exponential rate in Eq.~\eqref{growth}.} For all mobile nodes, $i\in\Theta$, we chose~$d^i_{k\rightarrow k+1}$ and~$\theta^i_{k\rightarrow k+1}$ in Eq.~\eqref{1} as random variables with uniform distributions over the intervals of $[0~5]$ and~$[0~2\pi]$, respectively. Each agent can only communicate with the nodes within its communication radius, which is set to $r=2$ for all simulations. \textcolor{black}{For each simulation, we assume exactly one fixed anchor, i.e., $M=1$.}
In all simulations, we set~$\alpha_k=0.2$ to ensure that the agents do not completely forget the past information and~$\alpha=0.1$ to guarantee an adequate contribution from the anchor(s).

Fig.~\ref{f1-2} (Left) shows the random trajectories of~$N=3$ mobile agents \textcolor{black}{and~$M=1$ anchor}} for the first~$25$ iterations. To characterize the convergence, we choose the $2$-norm of the error vector normalized with respect to the dimensions of the region as follows:
{\begin{equation}
{\Vert{\mb{e}_k}\Vert}_2=\frac{1}{2}\left({{\left(\frac{\mb{x}_k-\mb{x}_k^\ast}{\Delta x}\right)}^2+{\left(\frac{\mb{y}_k-\mb{y}_k^\ast}{\Delta y}\right)}^2}\right)^{\frac{1}{2}},
\end{equation}}in which \textcolor{black}{$\Delta x=20$ and $\Delta y=20$ denote the lengths of the region, which is a square in this case.}
\textcolor{black}{The division by $\Delta x$ and~$\Delta y$ is done to obtain a normalized error. It can be verified that the maximum error with this normalization is $\frac{1}{\sqrt{2}}$.}  
The algorithm converges when~${\| \mb{e}_k \|_{2}} \rightarrow 0$. As shown in Fig.~\ref{f1-2} (Right), the localization algorithm, Eq.~\eqref{eq1}, tracks the true agent locations.
\begin{figure}[!h]
\centering
\subfigure{\includegraphics[height=1.5in]{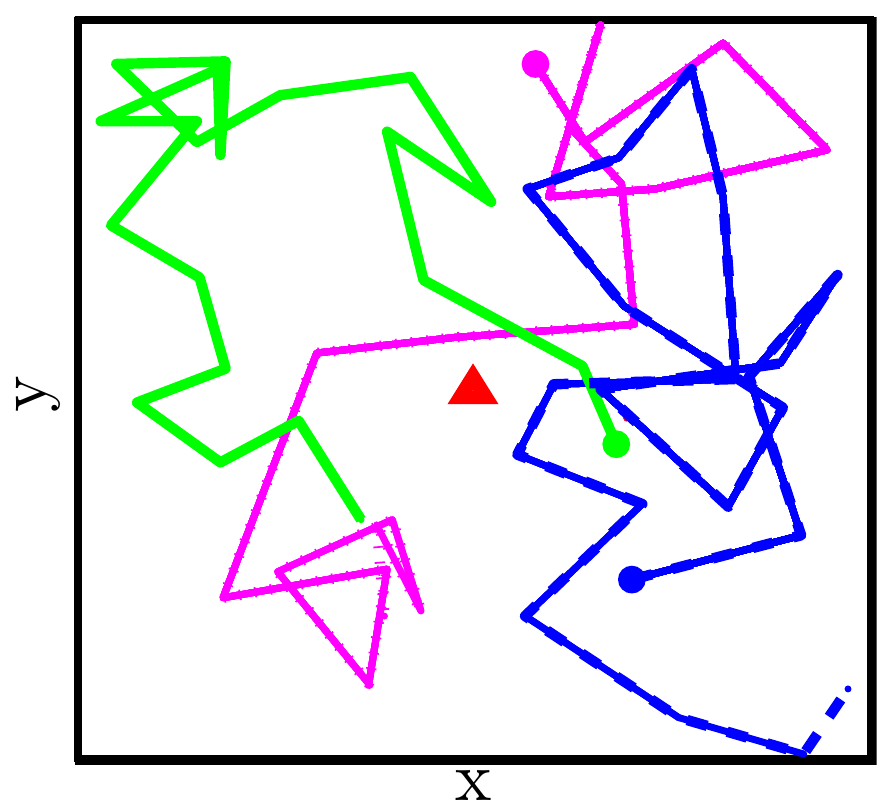}}
\hspace{0.1cm}
\subfigure{\includegraphics[height=1.6in,width=1.67in]{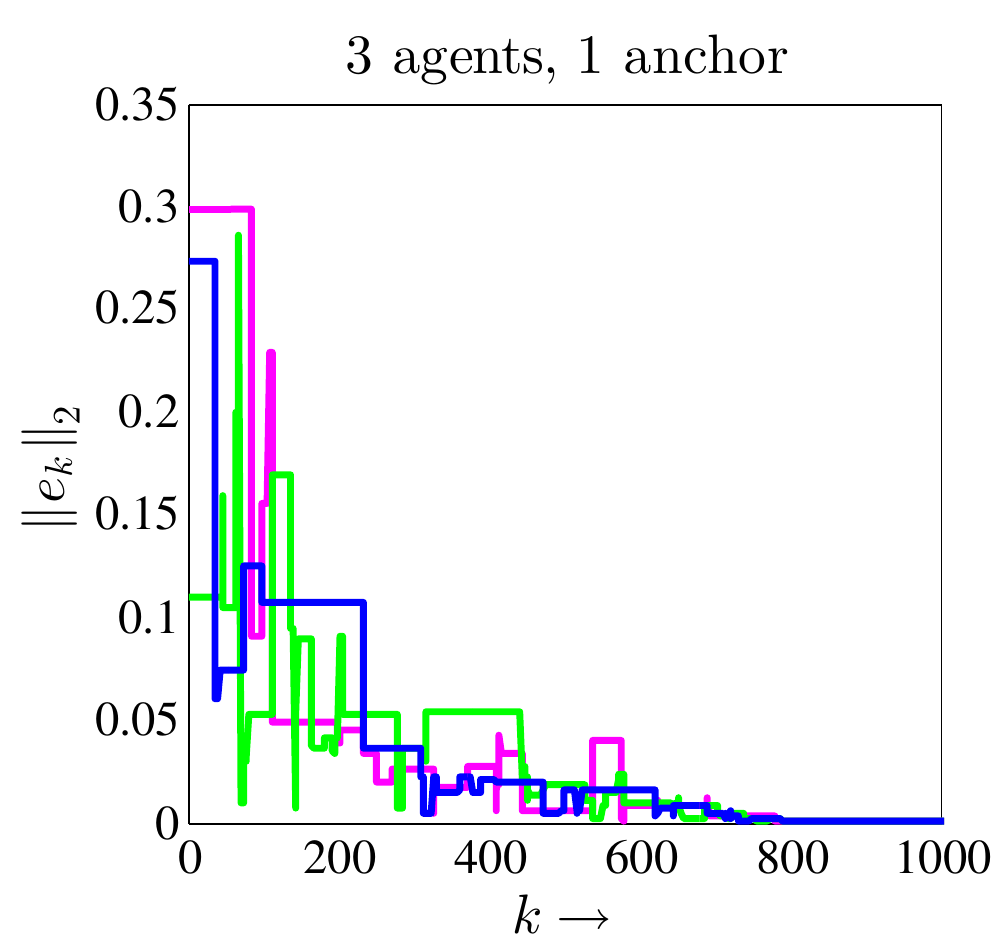}}
\caption{A network of~$3$ mobile agents \textcolor{black}{and~$1$ anchor}; \textcolor{black}{(Left) Motion model; (Right) Convergence; red triangle indicates anchor and filled circles show the initial positions of the agents.}}
\label{f1-2}
\end{figure}
The simulation results for a network of \textcolor{black}{$10$ and $100$} mobile agents and one mobile anchor is illustrated in Fig.~\ref{f10-2} (Left) and (Right), respectively.
\begin{figure}[!h]
	\centering
	\subfigure{\includegraphics[height=1.67in,width=1.67in]{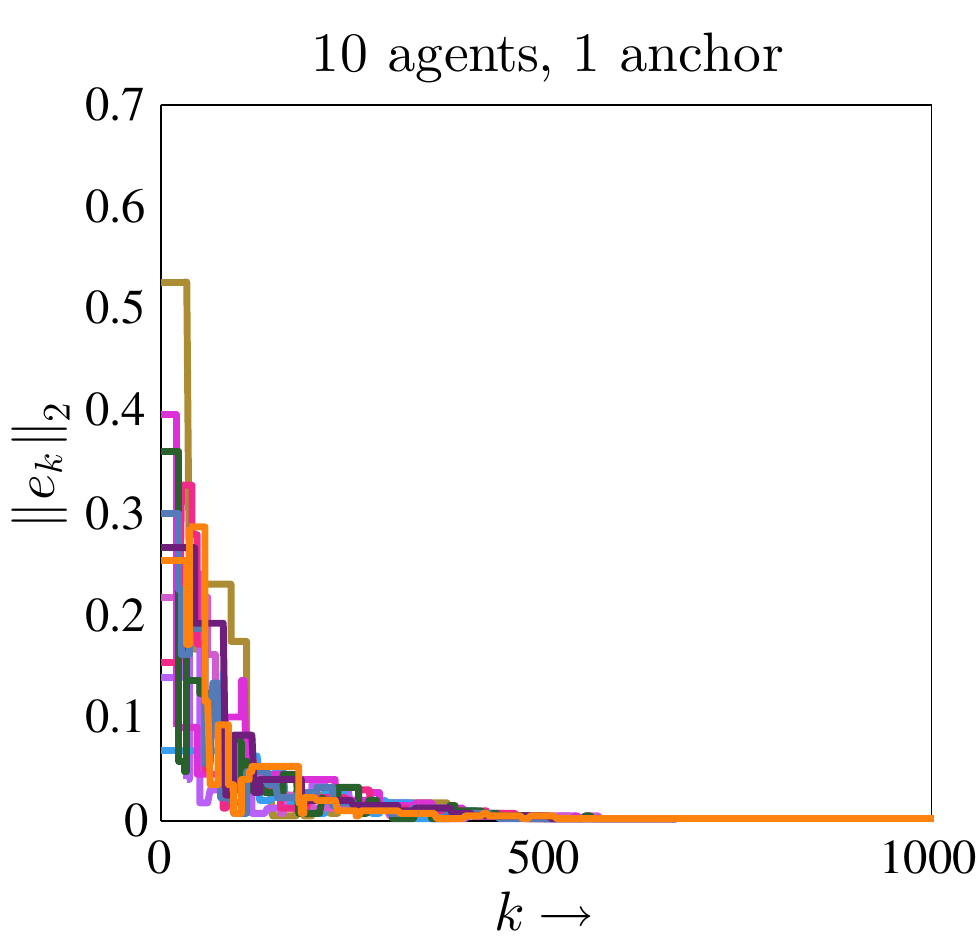}}
	\hspace{0.1cm}
	\subfigure{\includegraphics[height=1.67in,width=1.67in]{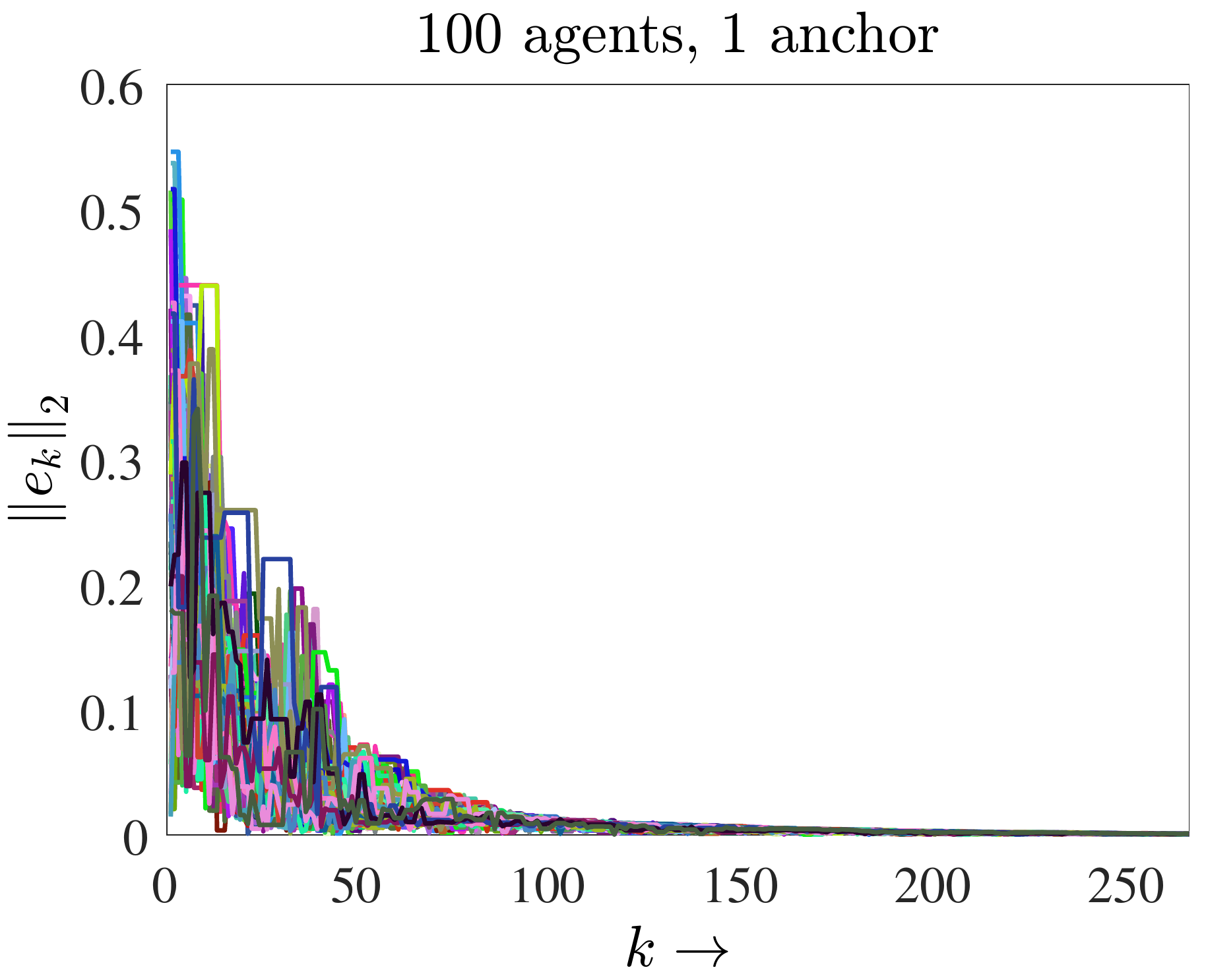}}
	\caption{\textcolor{black}{Convergence of a network with \textcolor{black}{one anchor}; (Left)~$10$ mobile agents (Right)~$100$ mobile agents.}}
	\label{f10-2}
\end{figure}
%
%

Before we examine the effects of noise on the proposed localization algorithm, let us study the case where there is no anchor in the network, i.e., Eq~\eqref{32} is not satisfied;
\textcolor{black}{In most} localization algorithms, $m+1$ anchors are required to localize an agent with unknown location in~$\mathbb{R}^{m}$ without ambiguity. However, when the nodes are mobile, \textit{exactly one anchor} is sufficient to inject valuable information and \textcolor{black}{the motion of the agents} provides the remaining degrees of freedom. When there is no anchor in the network, all updates are stochastic, and we get $\mb{e}_{{k+1}} = {\bf{P}}_{{k}}\mb{e}_{{k}}$, with $\rho({\bf{P}_k}) = 1,\forall k$, \textcolor{black}{where $\rho(.)$ denotes the spectral radius}, i.e., a neutrally stable system, which leads to a bounded steady-state error in location estimates. In other words, relative locations are tracked because the motion removes certain ambiguity from the locations. However, absolute tracking without any ambiguity requires at least one anchor. This is illustrated in~Fig.~\ref{figgg} for a network of $M=4$ agents and no anchor.

\begin{figure}[!h]
	\centering
	\subfigure{\includegraphics[height=1.67in,width=1.67in]{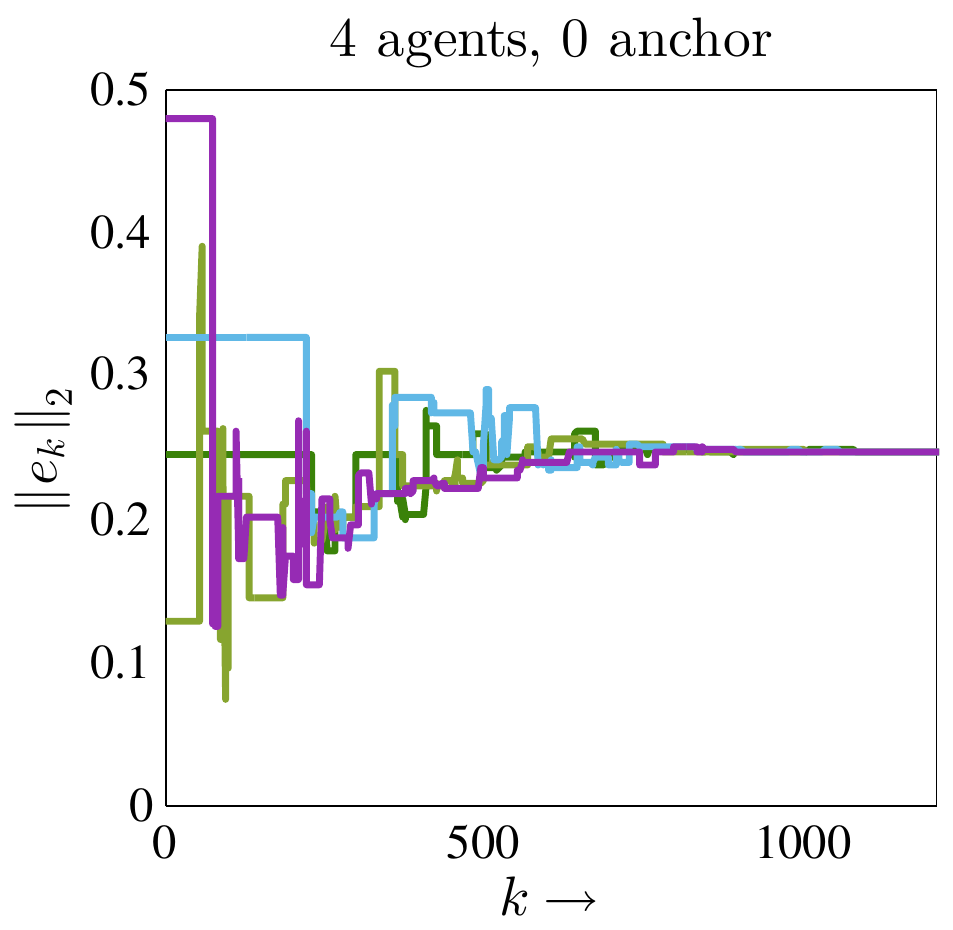}}
	\hspace{0.1cm}
	\subfigure{\includegraphics[height=1.67in,width=1.67in]{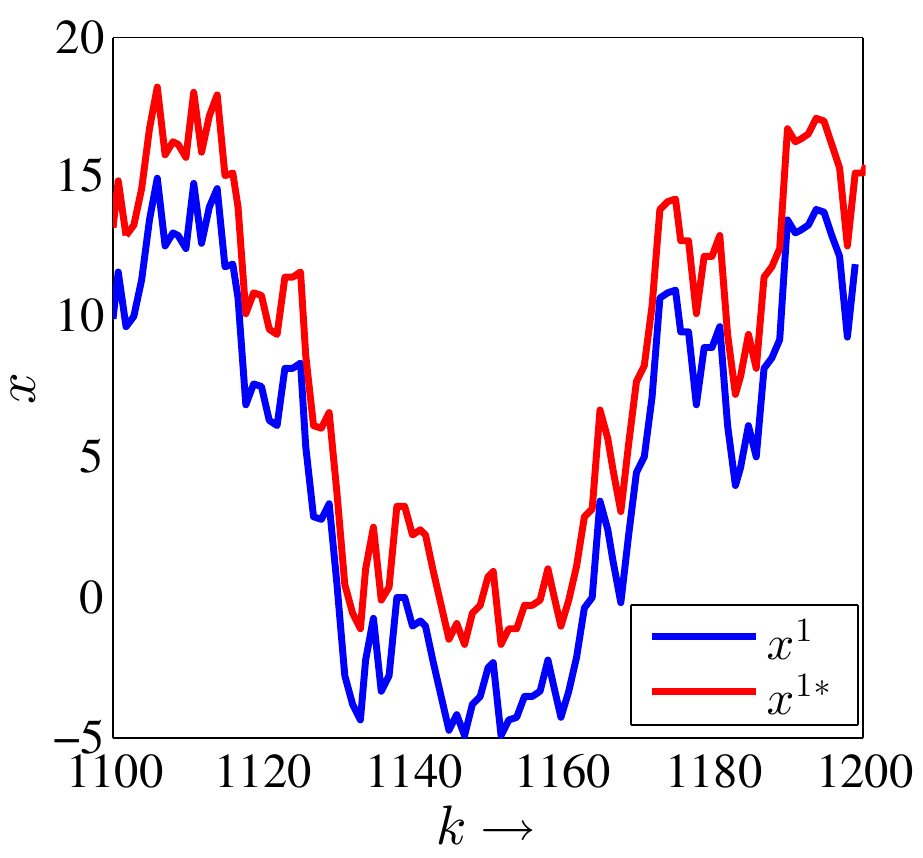}}
	\caption{Steady state error for a network of $4$ mobile agents and no anchor; (Left) $2$-norm (Right) Red and blue curves represent the exact and estimated values of~$x$ coordinate, for the last~$100$ iterations of a simulation.}
	\label{figgg}
\end{figure}

\vspace{-0cm}
\subsection{Localization with noise}\label{SEC7-B}
To examine the effects of noise on the proposed localization algorithm, we let the noise on the motion of agent $i \in \Theta$ at time $k$,  $n^i_{k}$ to be $\pm 1\%$ of the magnitude of the motion vector,~${\widetilde{\mb{x}}_{k}^i}$, and the noise on distance measurements, $r^{ij}_k$ to be $\pm 10\%$ of the actual distance between nodes $i$ and $j$,~$d^{ij}_k$. As shown in Fig.~\ref{f11-3} (Left) this amount of noise leads to an \textcolor{black}{unbounded error}, which is due to incorrect inclusion test results and the continuous location drifts because of the noise on the distance measurements and the noise on motion, respectively. However, by modifying the algorithm according to Section~\ref{noise}, it can be seen in Fig.~\ref{f11-3} (Right) that the normalized localization error is reduced to less than $5\%$ for one simulation.
\begin{figure}[!h]
	\centering
	\subfigure{\includegraphics[height=1.67in,width=1.67in]{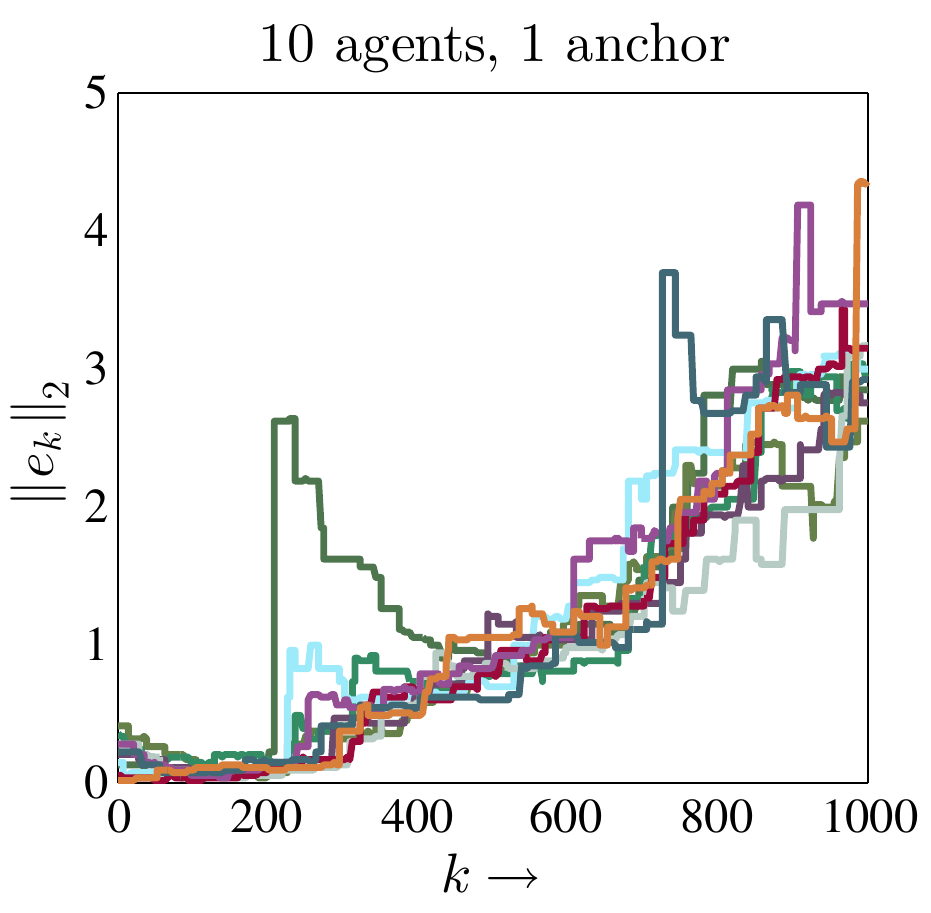}}
	\hspace{0.1cm}
	\subfigure{\includegraphics[height=1.67in,width=1.67in]{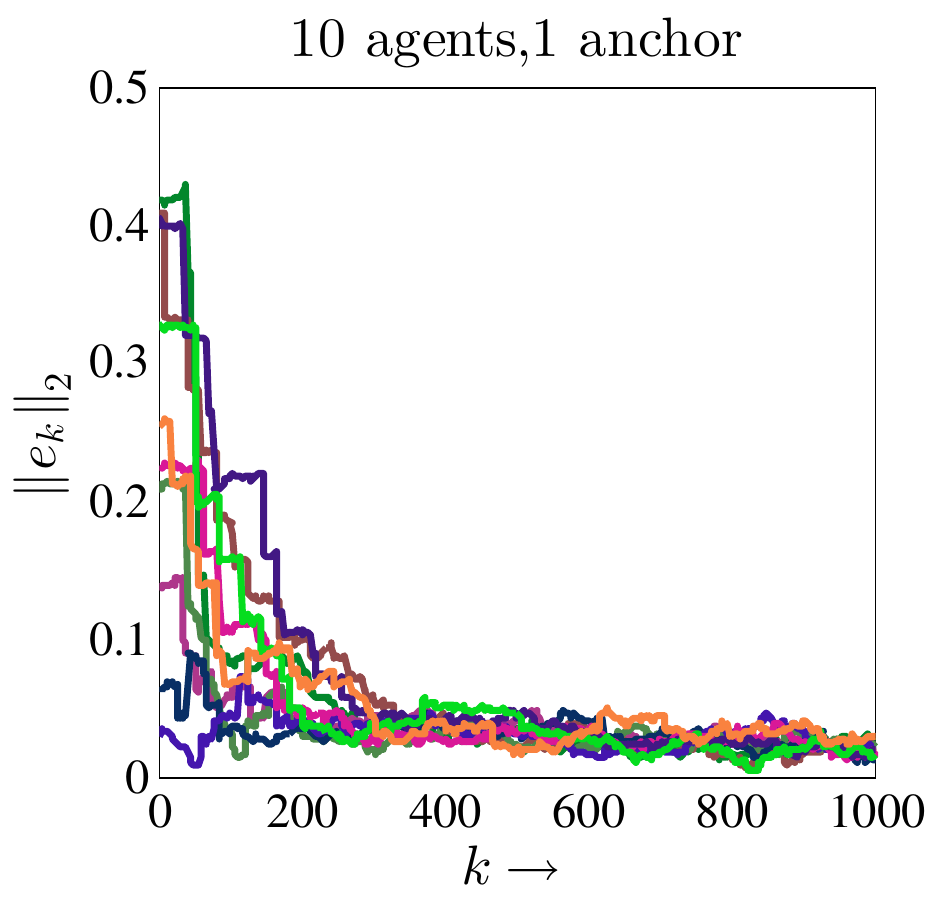}}
	\caption{\textcolor{black}{Effect of noise on the convergence of a network of~$10$ mobile agents and~$1$ mobile anchor with $\pm 10\%$ noise on distance measurements and $\pm 1\%$ noise on the motion; (Left) Original algorithm (Right) Modified algorithm.}}
	\label{f11-3}
\end{figure}
%
%
%

\noindent Finally, in Fig.~\ref{f14} we provide  $20$ Monte Carlo simulations for a network of $M=1$ mobile anchor and $N=10$ \textcolor{black}{and~$N=20$ mobile agents} in presence of noise.
\begin{figure}[!h]
	\centering
	\subfigure{\includegraphics[height=1.67in,width=1.67in]{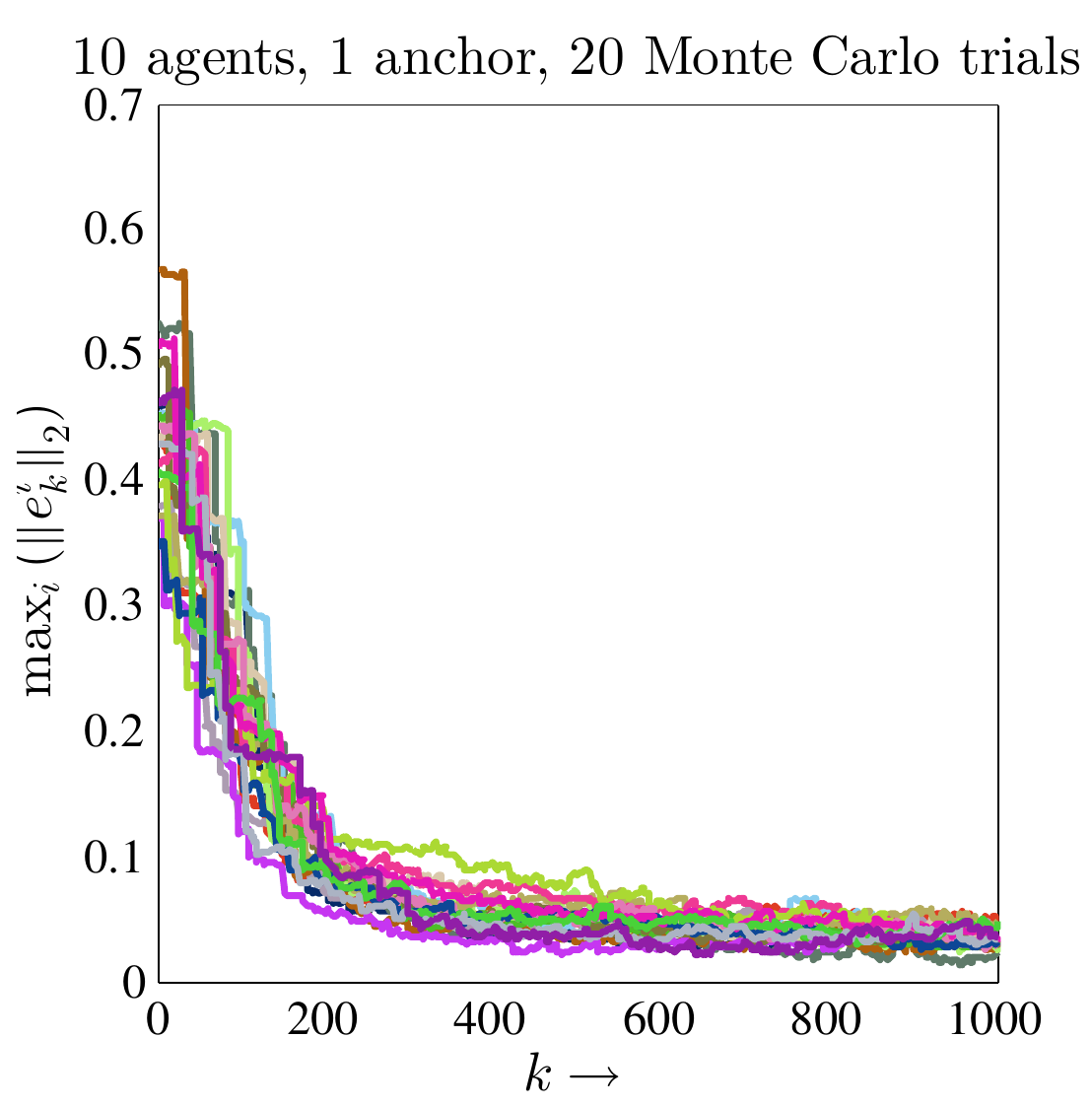}}
	\hspace{0.1cm}
	\subfigure{\includegraphics[height=1.67in,width=1.67in]{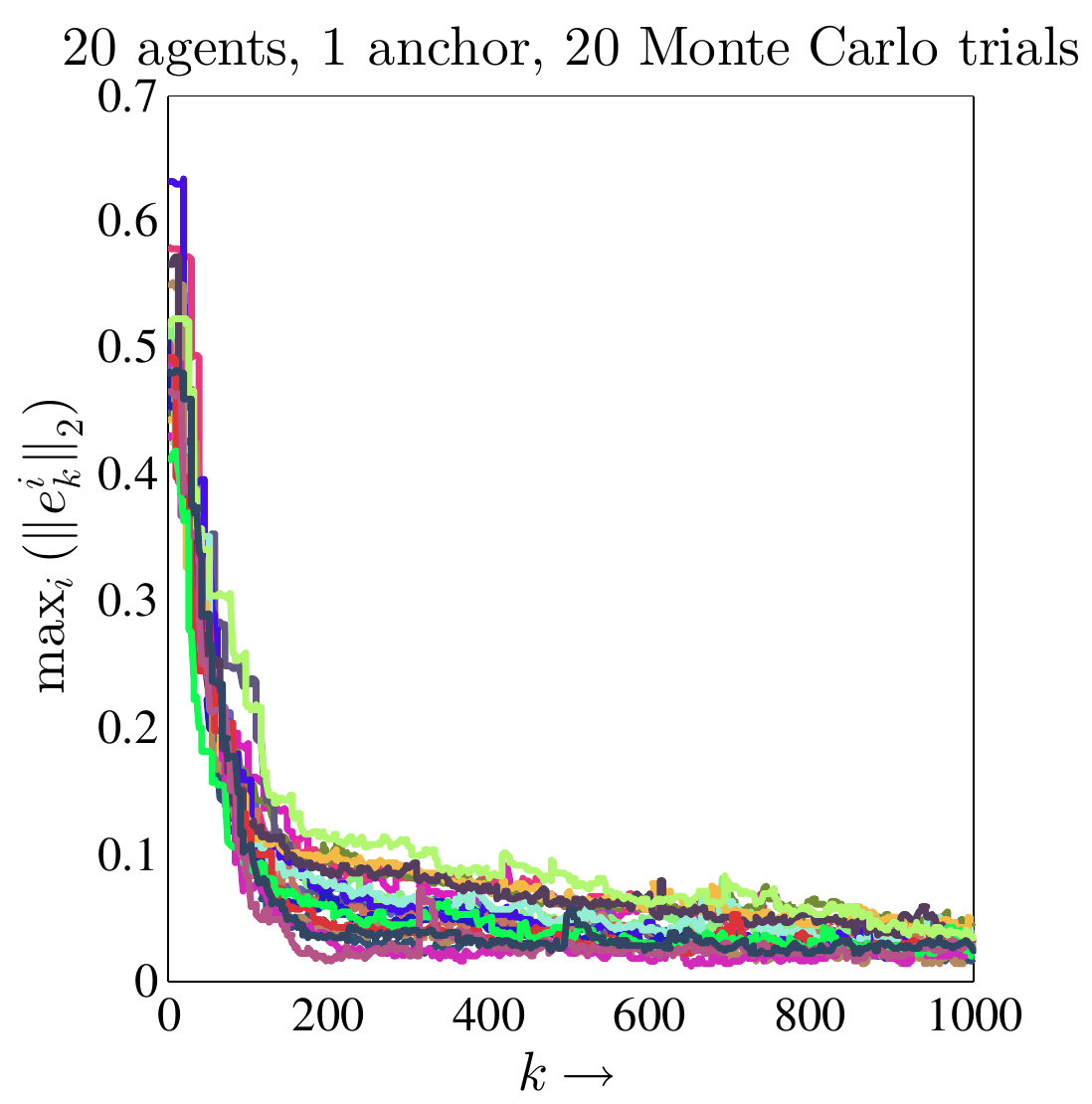}}
	\caption{\textcolor{black}{$20$ Monte Carlo trials; $\pm 10\%$ noise on distance measurements, $\pm 1\%$ noise on the motion; (Left) $N=10,~M=1$ (Right) \textcolor{black}{$N=20,~M=1$.}}}
	\label{f14}
\end{figure}

\subsection{Performance evaluation}\label{eval}
\textcolor{black}{We now evaluate the performance of our algorithm in contrast with some well-known localization methods; MCL~\cite{hu2004localization}, MSL*~\cite{rudafshani2007localization}, Dual MCL~\cite{stevens2007dual}, and Range-based SMCL~\cite{dil2006range}.
Please refer back to Section~\ref{sec1} for a brief description of these methods.
In Fig.~\ref{f16}, we compare the localization error in the Virtual Convex Hull (VCH) algorithm with MCL, MSL*, and Range-based SMC. 
As shown in Table~\ref{tbl1}, in these algorithms node density,~$n_d$, and anchor (seed) density,~$n_s$, denote as the average number of nodes and anchors in the neighborhood of an agent, respectively. Total number of agents and anchors can therefore be determined by knowing these densities as well as the area of the region. 
	\begin{figure}[!h]
		\centering
		\subfigure{\includegraphics[height=2.15in,width=2.15in]{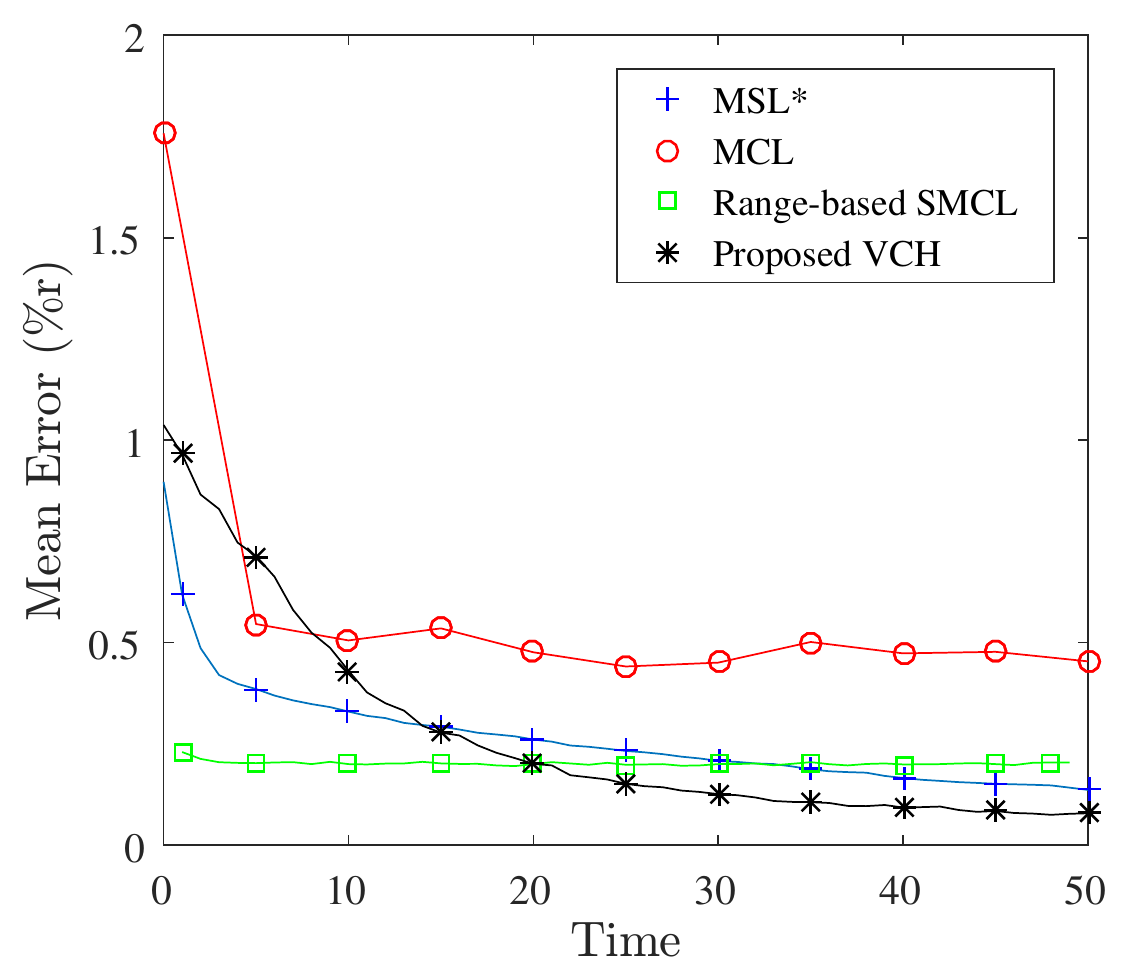}}
		\caption{Accuracy comparison}
		\label{f16}
	\end{figure}
We consider $N=20$ agents and $M=2$ anchors, and use the same metric as used in~\cite{hu2004localization,rudafshani2007localization,dil2006range}, i.e., the location error as a percentage of the communication range. Each data point in Fig.~\ref{f16} is computed by averaging the results of $20$ simulation experiments. We keep the other parameters the same as described earlier in this section. With high measurement noise, i.e., in the presence of $10\%$ noise on the range measurements and $1\%$ noise on the motion, our algorithm outperforms MCL, MSL* and Range-based SMCL after $20$ iterations. Clearly, the localization error in our algorithm decreases as the amount of noise decreases, and our algorithm converges to the exact agent locations in the absence of noise.}
\textcolor{black}{Table~\ref{tbl1} summarizes the performance of the proposed VCH algorithm in comparison to the above methods.}
\begin{table}[!h]
	\centering
	\subfigure{\includegraphics[height=1.7in]{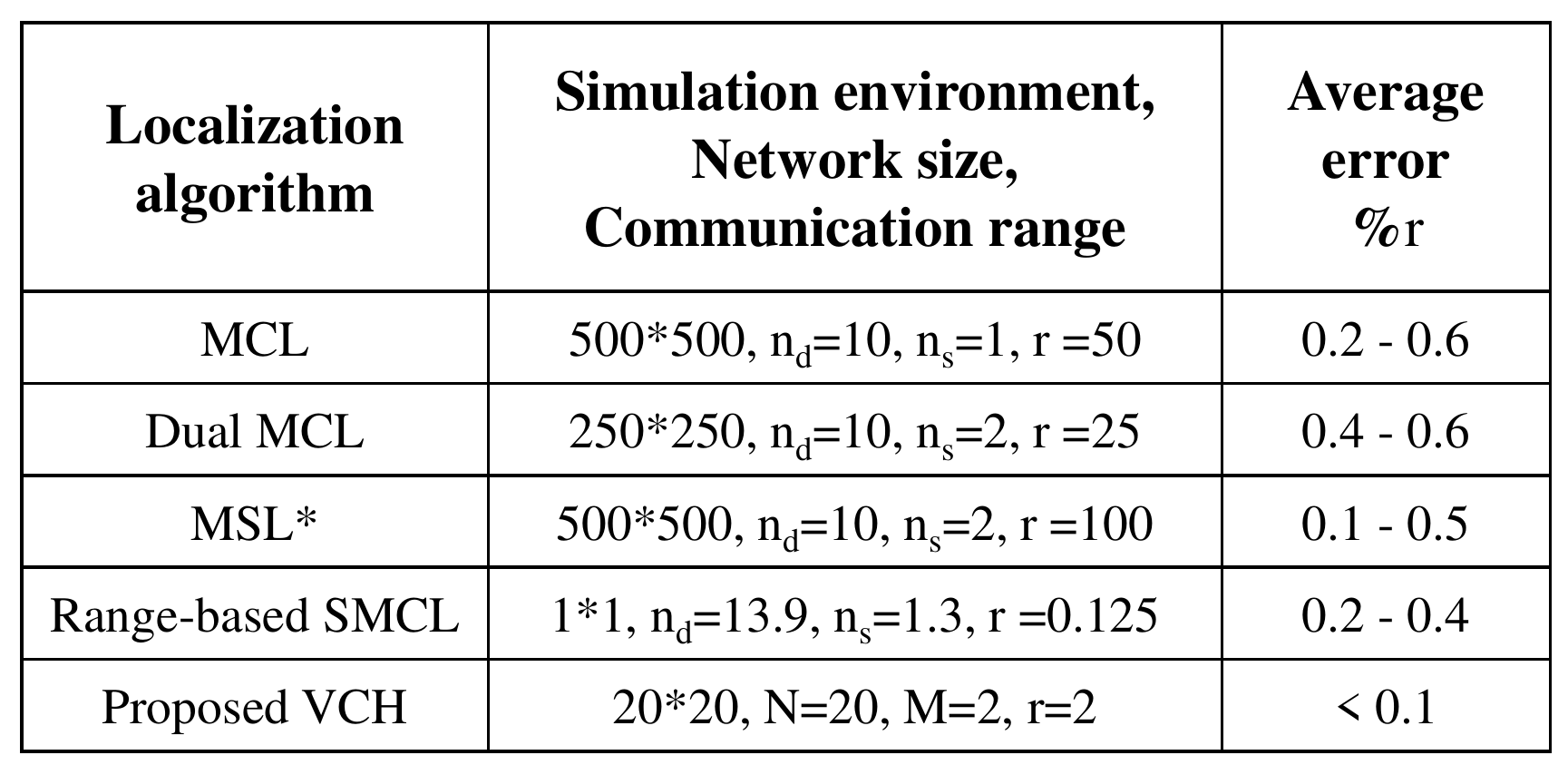}}
	\caption{Comparative performance of localization algorithms}
	\label{tbl1}
\end{table}

\section{Conclusions}\label{sec9}
In this paper, we provide a distributed algorithm to localize a network of mobile agents moving in a bounded region of interest. Assuming that each agent knows a noisy version of its motion and the distances to the nodes in its communication radius, we provide a geometric framework, which allows an agent to keep track of the distances to any previously \textcolor{black}{visited nodes}, and find the distance between a neighbor and any virtual location where it has exchanged information with other nodes in the past. 
Since agents are mobile, they may not be able to find neighbors to perform distributed updates at any time.
To avoid this issue, we introduce the notion of a virtual convex hull, which forms the basis of a range-based localization algorithm in mobile networks. We abstract the algorithm as an LTV system with (sub-) stochastic matrices, and show that it converges to the true locations of agents under some mild regularity conditions on update weights. We further show that exactly one anchor suffices to localize an arbitrary number of mobile agents when each agent is able to find appropriate (virtual) convex hulls. We evaluate the performance of the algorithm \textcolor{black}{in presence of noise and provide modifications to the proposed algorithm to address noise on motion and on distance measurements.}

\appendices
\vspace{-0.5cm}
\textcolor{black}{\section{Table of parameters}\label{tbl}
Table~\ref{figz} summarizes the important notation. 
	\begin{table*}
		\centering
		\includegraphics[width=190mm]{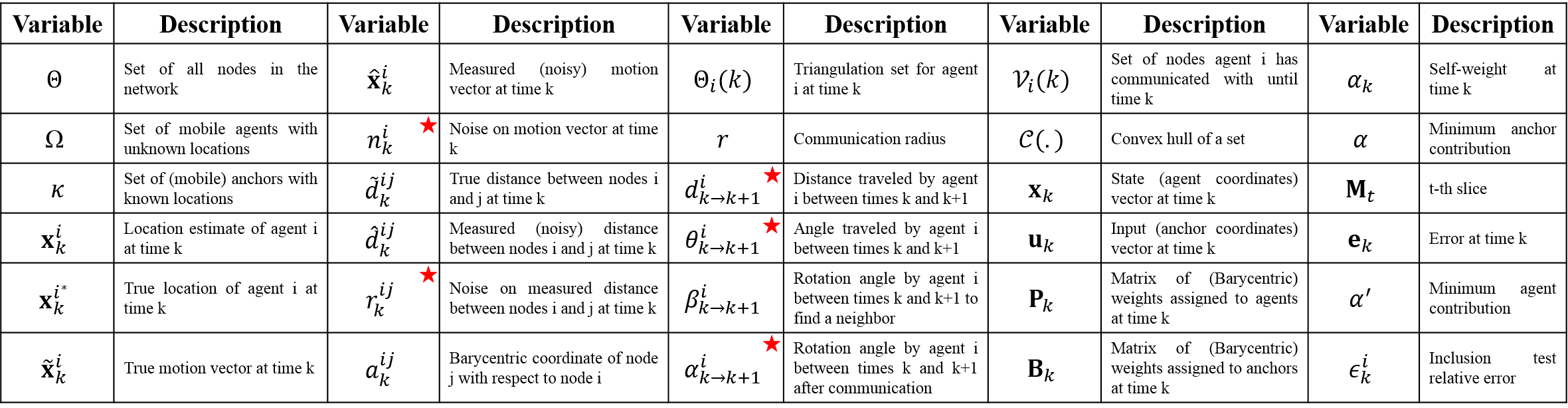}
		\caption{\textcolor{black}{Parameters and descriptions; random parameters are distinguished by \textcolor{red}{$\star$}}.}
		\label{figz}
	\end{table*}}

\section{Proof of Lemma 1}\label{pr1}
\begin{proof}
	Consider the triangle, \textcolor{black}{whose vertices are at}~$\mb{x}_{k_j}^{i\ast}$,~$\mb{x}_{k_j+1}^{i\ast}$ and~$\mb{x}_{k_j}^{j\ast}$. To find the current distance of agent~$i$, from the position of agent~$j$ at time~${k}_{j}$, we can use the \textit{law of cosines}, which connects the length of an unknown side of a triangle to the lengths of the two other sides and the angle opposite to the unknown side. This triangle is depicted in Fig.~\ref{f2}~(c), in which the two known side lengths are~$\textcolor{black}{\widetilde{d}_{k_j}^{ij}}$ and~${d_{{k_j\rightarrow {k_j}+1}}^{i}}$. Thus, knowing the angle between these sides of the triangle,~$\alpha^i_{k_j\rightarrow {k_j}+1}$, the length of the third side can be determined according to the following
	\begin{eqnarray}\label{5}
	(\widetilde{d}_{k_j+1}^{ij})^2 &=& (\textcolor{black}{\widetilde{d}_{k_j}^{ij}})^{2} + (\textcolor{black}{d_{{k_j\rightarrow k_{j+1}}}^{i}})^{2} \\\nonumber
	&-& 2{{\textcolor{black}{\widetilde{d}_{k_j}^{ij}}}{\textcolor{black}{d_{{k_j\rightarrow k_{j+1}}}^{i}}}\cos(\alpha^i_{k_j})},
	\end{eqnarray}
	which corresponds to Eq.~\eqref{3-1}, and completes the proof.
\end{proof}

\section{Proof of Lemma 2}\label{pr2}
\begin{proof}
	We illustrate this procedure in Fig.~\ref{fig1-1},
	where all known/measured distances and angles are distinguished from the unknowns by bold lines and bold arcs, respectively. 
	After agent $i$ makes contact and exchange information with node~$j$ at time $k_j$, it moves the distance of $d_{k_j\rightarrow k_j+1}^i$ to the new location,~$\mb{x}_{k_j+1}^{i\ast}$. At this point, agent $i$ can (use the law of cosines to) find $\widetilde{d}_{k_j+1}^{ij}$ as a function of \textcolor{black}{$\widetilde{d}_{k_j}^{ij}$},~\textcolor{black}{$d^{i}_{k_j\rightarrow k_j+1}$}, and $\alpha_{k_j}^{i}$, which are all known to this agent, i.e.,~$\widetilde{d}_{k_j+1}^{ij}=f(\textcolor{black}{\widetilde{d}_{k_j}^{ij}},\textcolor{black}{d_{k_j\rightarrow k_j+1}^{i}},\alpha_{k_j}^{i})$. Knowing $\widetilde{d}_{k_j+1}^{ij}$, agent $i$ can then find~\textcolor{black}{$\phi_1=g(\widetilde{d}_{k_j}^{ij}, d_{k_j\rightarrow k_j+1}^{i},\widetilde{d}_{k_j+1}^{ij})$}. 
	Note that given the three sides of a triangle, each angle can be computed.
	Since agent $i$ cannot find any neighbor at time $k_j+1$, we have~$\beta_{k_j+1}^i=2\pi$. Agent $i$ then changes its direction by $\alpha_{k_j+1}^{i}$ and travels the distance of $d_{k_j+1\rightarrow k_\ell}^i$ to the new location,~$\mb{x}_{k_\ell}^{i\ast}$, where $k_\ell=k_j+2$. At this point, agent $i$ can find~$\widetilde{d}_{k_\ell}^{ij}=f(\widetilde{d}_{k_j+1}^{ij} ,{{{d}}}_{k_j+1\rightarrow k_\ell}^{i},\phi_2)$, note that $\phi_2=\pi-\alpha_{k_j+1}^{i}+\phi_1$. Agent~$i$ can then find~\textcolor{black}{$\phi_3=g(\widetilde{d}_{k_\ell}^{ij},{{{d}}}_{k_j+1\rightarrow k_\ell}^{i},\widetilde{d}_{k_j+1}^{ij})$} at time~$k_\ell$. 
	Since agent $i$ finds node~$\ell$ within the communication radius at time $k_\ell$, it changes its direction by $\beta_{k_\ell}^{i}$ in order to make a contact with agent $\ell$.
	Finally, knowing \textcolor{black}{$\widetilde{d}_{k_\ell}^{i\ell}$},~$\widetilde{d}_{k_\ell}^{ij}$, and,~$\phi_4=\angle(\widetilde{d}_{k_\ell}^{ij},\textcolor{black}{{\widetilde{d}_{k_\ell}^{i\ell}}})=\pi-\phi_3-\beta_{k_\ell}^{i}$, agent $i$ can
	find the distance between~$\mb{x}_{k_j}^{j\ast}$ and~$\mb{x}_{k_\ell}^{\ell\ast}$ according to Eq.~\eqref{14}.
\end{proof}

\vspace{-7mm}
\textcolor{black}{\section{Proof of Lemma 3}\label{pr3}
\begin{proof}
Consider a network of $4$ agents,~$\{i,j,m,l\}$, and let us focus on finding a triangulation set for agent $i$. 
	\begin{figure}[!h]
		\centering
		\includegraphics[width=1.5in]{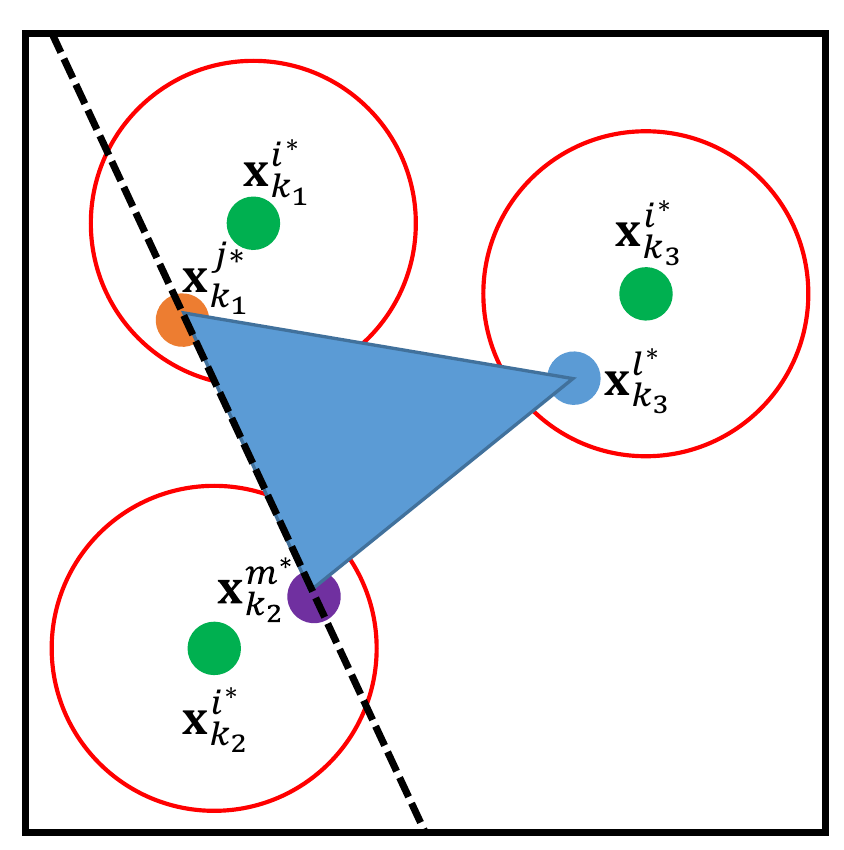}
		\caption{Existence of a virtual convex hull; Agent $i$ can perform an update if it lies inside the blue triangle at any time $k>k_3$.}
		\label{figlem3}
	\end{figure}
Suppose all agents are moving in an $n$ by~$n$ region, and the communication radius for each agent is $r$. In order for agent $i$ to communicate with another agent, say $j$, at time $k_1$, agent $j$ has to lie inside a circle with radius~$r$ centered at agent $i$'s location,~${\bf{x}}_{k_1}^{i\ast}$, see Fig.~\ref{figlem3}. The probability of such event at time $k_1$, is therefore given by $\mathbb{P}(j \in \mathcal{V}_i(k_1))=\frac{\pi r^2}{n^2}$\textcolor{black}{\footnote{\textcolor{black}{This expression assumes that agents are uniformly distributed over the region at any given time. This assumption can be justified by considering random initial deployment and motion. However, the proof follows as long as there is a nonzero probability for each agent to communicate with other agent(s) in the network.}}}. 
Similarly, the probability that agent $i$ communicates with agent $m$ at a later time,~$k_2 > k_1$, and with agent $l$, at time~$k_3 >k_2 >k_1$ can be given by~$\mathbb{P}(m \in \mathcal{V}_i(k_2))=\frac{\pi r^2}{n^2}$ and~$\mathbb{P}(l \in\mathcal{V}_i(k_3))=\frac{\pi r^2}{n^2}$, respectively. \textcolor{black}{Note that the probability of agent $m$ being exactly at ${\bf{x}}_{k_1}^{j\ast}$ at time $k_2$ is zero}, and therefore we can draw a virtual line (the dotted line in Fig.~\ref{figlem3}) between the locations at which agent~$i$ has communicated with agents $j$ and $m$,~i.e., between~${\bf{x}}_{k_1}^{j\ast}$ and ${\bf{x}}_{k_2}^{m\ast}$. Similarly the probability of agent $l$ lying on the aforementioned line at time $k_3$ is zero. Thus, the locations where agent $i$ can meet the other three agents at different times with nonzero probability, form a virtual convex hull. Therefore, the three agents $j$, $m$ and $l$ can pass the inclusion test for agent $i$ at any given time $k > k_3$, if agent $i$ lies inside the corresponding virtual triangle. The probability of such event at any given time, $k > k_3$, is given by
\begin{equation*}
\mathbb{P}(i\in\mc{C}({\bf{x}}_{k_1}^{j\ast},{\bf{x}}_{k_2}^{m\ast},{\bf{x}}_{k_3}^{l\ast}))=\frac{\mbox{area}(\triangle({\bf{x}}_{k_1}^{j\ast},{\bf{x}}_{k_2}^{m\ast},{\bf{x}}_{k_3}^{l\ast})}{n^2},
\end{equation*} 
which corresponds to a non-zero probability for agent $i$ to find a triangulation set.
\end{proof}}

\textcolor{black}{\section{Pseudocode}\label{psuedo}
	\begin{algorithm}
		\caption{Localize $N$ agents in $\mathbb{R}^{2}$ in the presence of $M$ anchors with precision $p$}
		\begin{algorithmic} 
			\REQUIRE $M \geq 1~\AND~N+M \geq 4$\\
			\STATE $k \leftarrow 0$
			\STATE ${\bf{x}}_0 \leftarrow \mbox{random initial coordinates}$
			\FOR {$i=1$ to $N$}
			\STATE ${\mathcal{V}}_i(0) = \emptyset$
			\ENDFOR
			\WHILE{$k<$ \textcolor{black}{termination criterion\footnotemark}}
			\STATE \textcolor{black}{$k \leftarrow k+1$}
			\FOR{$i=1$ to $N$}
			\STATE ${\mathcal{V}}_i(k) \leftarrow$ nodes in the communication radius of agent~$i$ at time $k$
			\IF{$0\leq\vert{\mathcal{V}}_i(k)\vert < 3$}
			\STATE do not update
			\ELSE
			\STATE perform inclusion test on (all possible combinations of) $3$ neighbors
			\IF{no triangulation set found}
			\STATE do not update
			\ELSE
			\STATE update location according to Eq.~\eqref{18}
			\ENDIF
			\ENDIF
			\ENDFOR
			\ENDWHILE
		\end{algorithmic}
	\end{algorithm}}
	\footnotetext{\textcolor{black}{The termination criterion can be designed according to the number of iterations typically needed given the size, mobility, models, and noise parameters, as evident from the simulation figures in Section~\ref{sec7}.}}

\section{Convex hull inclusion test}\label{incl}
In $\mathbb{R}^2$, a convex hull inclusion test is as follows 
{\small\begin{eqnarray}\label{convEQ}
i\in\mathcal{C}(\Theta_i(k)),\qquad \mbox{if } \sum_{j\in\Theta_i(k)}A_{\Theta_i(k)\cup\{i\}\setminus j} = A_{\Theta_i(k)},\\
i\notin\mathcal{C}(\Theta_i(k)),\qquad \mbox{if } \sum_{j\in\Theta_i(k)}A_{\Theta_i(k)\cup\{i\}\setminus j} > A_{\Theta_i(k)},
\end{eqnarray}}
\begin{figure}[!h]
	\centering
	\includegraphics[width=55mm]{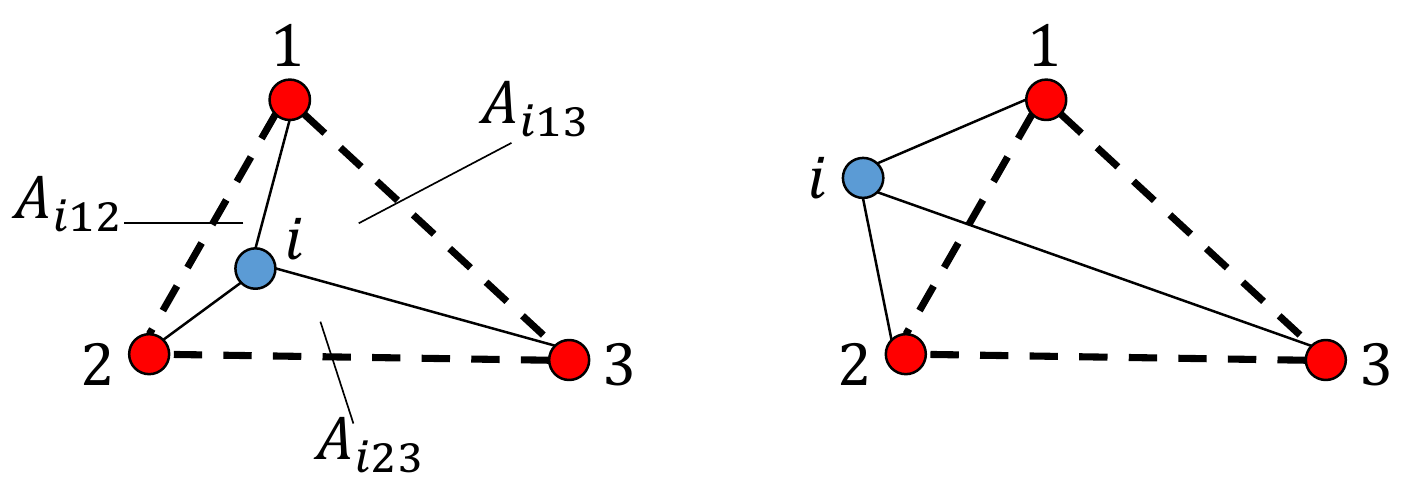}
	\caption{$\mbb{R}^2$: (Left) Robot~$i$ lies inside the triangle formed by the neighboring nodes,~$1$,~$2$ and~$3$
		(Right) Inclusion test is not passed.} 
	\label{fig1}
\end{figure}
in which~$``\setminus"$ denotes the set difference, $\mathcal{C}(\cdot)$ is the convex hull, and~$A_{\Theta_i(k)}$ represents the area of~$\mathcal{C}(\Theta_i(k))$, see Fig.~\ref{fig1}, and can be computed by using the Cayley-Menger determinant,~\cite{sippl1986cayley}. Thus,~$A_{\Theta_i(k)\cup\{i\}\setminus j}$ is the area of the set~$\Theta_i(k)$ with node $i$ added and node $j$ removed.

\bibliographystyle{IEEEtran}
\bibliography{bibliography}

\begin{thebibliography}{10}
\providecommand{\url}[1]{#1}
\csname url@samestyle\endcsname
\providecommand{\newblock}{\relax}
\providecommand{\bibinfo}[2]{#2}
\providecommand{\BIBentrySTDinterwordspacing}{\spaceskip=0pt\relax}
\providecommand{\BIBentryALTinterwordstretchfactor}{4}
\providecommand{\BIBentryALTinterwordspacing}{\spaceskip=\fontdimen2\font plus
\BIBentryALTinterwordstretchfactor\fontdimen3\font minus
  \fontdimen4\font\relax}
\providecommand{\BIBforeignlanguage}[2]{{%
\expandafter\ifx\csname l@#1\endcsname\relax
\typeout{** WARNING: IEEEtran.bst: No hyphenation pattern has been}%
\typeout{** loaded for the language `#1'. Using the pattern for}%
\typeout{** the default language instead.}%
\else
\language=\csname l@#1\endcsname
\fi
#2}}
\providecommand{\BIBdecl}{\relax}
\BIBdecl

\bibitem{estrin2001instrumenting}
D.~Estrin, L.~Girod, G.~Pottie, and M.~Srivastava, ``Instrumenting the world
  with wireless sensor networks,'' in \emph{Acoustics, Speech, and Signal
  Processing, 2001. Proceedings.(ICASSP'01). 2001 IEEE International Conference
  on}, vol.~4.\hskip 1em plus 0.5em minus 0.4em\relax IEEE, 2001, pp.
  2033--2036.

\bibitem{pottie2000wireless}
G.~J. Pottie and W.~J. Kaiser, ``Wireless integrated network sensors,''
  \emph{Communications of the ACM}, vol.~43, no.~5, pp. 51--58, 2000.

\bibitem{estrin1999next}
D.~Estrin, R.~Govindan, J.~Heidemann, and S.~Kumar, ``Next century challenges:
  Scalable coordination in sensor networks,'' in \emph{Proceedings of the 5th
  annual ACM/IEEE international conference on Mobile computing and
  networking}.\hskip 1em plus 0.5em minus 0.4em\relax ACM, 1999, pp. 263--270.

\bibitem{steere2000research}
D.~C. Steere, A.~Baptista, D.~McNamee, C.~Pu, and J.~Walpole, ``Research
  challenges in environmental observation and forecasting systems,'' in
  \emph{Proceedings of the 6th annual international conference on Mobile
  computing and networking}.\hskip 1em plus 0.5em minus 0.4em\relax ACM, 2000,
  pp. 292--299.

\bibitem{Pandey2016}
S.~Pandey and S.~Varma, ``A range based localization system in multihop
  wireless sensor networks: A distributed cooperative approach,''
  \emph{Wireless Personal Communications}, vol.~86, no.~2, pp. 615--634, 2016.

\bibitem{zhang2005range}
Y.~Zhang, W.~Wu, and Y.~Chen, ``A range-based localization algorithm for
  wireless sensor networks,'' \emph{Journal of Communications and Networks},
  vol.~7, no.~4, pp. 429--437, 2005.

\bibitem{7496598}
S.~A. Mageid, ``Autonomous localization scheme for mobile sensor networks in
  fading environments,'' in \emph{IEEE International Conference on Selected
  Topics in Mobile Wireless Networking (MoWNeT)}, April 2016, pp. 1--8.

\bibitem{hu2004localization}
L.~Hu and D.~Evans, ``Localization for mobile sensor networks,'' in
  \emph{Proceedings of the 10th annual international conference on Mobile
  computing and networking}.\hskip 1em plus 0.5em minus 0.4em\relax ACM, 2004,
  pp. 45--57.

\bibitem{dil2006range}
B.~Dil, S.~Dulman, and P.~Havinga, ``Range-based localization in mobile sensor
  networks,'' in \emph{Wireless Sensor Networks}.\hskip 1em plus 0.5em minus
  0.4em\relax Springer, 2006, pp. 164--179.

\bibitem{rudafshani2007localization}
M.~Rudafshani and S.~Datta, ``Localization in wireless sensor networks,'' in
  \emph{Information Processing in Sensor Networks, 2007. IPSN 2007. 6th
  International Symposium on}.\hskip 1em plus 0.5em minus 0.4em\relax IEEE,
  2007, pp. 51--60.

\bibitem{stevens2007dual}
E.~Stevens-Navarro, V.~Vivekanandan, and V.~W. Wong, ``Dual and mixture monte
  carlo localization algorithms for mobile wireless sensor networks,'' in
  \emph{IEEE Wireless Communications and Networking Conference}.\hskip 1em plus
  0.5em minus 0.4em\relax IEEE, 2007, pp. 4024--4028.

\bibitem{patwari2003relative}
N.~Patwari, A.~O. Hero, M.~Perkins, N.~S. Correal, and R.~J. O'dea, ``Relative
  location estimation in wireless sensor networks,'' \emph{Signal Processing,
  IEEE Transactions on}, vol.~51, no.~8, pp. 2137--2148, 2003.

\bibitem{thomas2005revisiting}
F.~Thomas and L.~Ros, ``Revisiting trilateration for robot localization,''
  \emph{Robotics, IEEE Transactions on}, vol.~21, no.~1, pp. 93--101, 2005.

\bibitem{roweis2000nonlinear}
S.~T. Roweis and L.~K. Saul, ``Nonlinear dimensionality reduction by locally
  linear embedding,'' \emph{Science}, vol. 290, no. 5500, pp. 2323--2326, 2000.

\bibitem{aspnes2006theory}
J.~Aspnes, T.~Eren, D.~K. Goldenberg, A.~S. Morse, W.~Whiteley, Y.~R. Yang,
  B.~D. Anderson, and P.~N. Belhumeur, ``A theory of network localization,''
  \emph{Mobile Computing, IEEE Transactions on}, vol.~5, no.~12, pp.
  1663--1678, 2006.

\bibitem{albowicz2001recursive}
J.~Albowicz, A.~Chen, and L.~Zhang, ``Recursive position estimation in sensor
  networks,'' in \emph{Network Protocols, 2001. Ninth International Conference
  on}.\hskip 1em plus 0.5em minus 0.4em\relax IEEE, 2001, pp. 35--41.

\bibitem{savarese2001location}
C.~Savarese, J.~M. Rabaey, and J.~Beutel, ``Location in distributed ad-hoc
  wireless sensor networks,'' in \emph{Acoustics, Speech, and Signal
  Processing, 2001. Proceedings.(ICASSP'01). 2001 IEEE International Conference
  on}, vol.~4.\hskip 1em plus 0.5em minus 0.4em\relax IEEE, 2001, pp.
  2037--2040.

\bibitem{vcapkun2002gps}
S.~{\v{C}}apkun, M.~Hamdi, and J.-P. Hubaux, ``Gps-free positioning in mobile
  ad hoc networks,'' \emph{Cluster Computing}, vol.~5, no.~2, pp. 157--167,
  2002.

\bibitem{sriv:02}
A.~Savvides, H.~Park, and M.~B. Srivastava, ``The bits and flops of the n-hop
  multilateration primitive for node localization problems,'' in \emph{Intl.
  Workshop on Sensor Networks and Applications}, Atlanta, GA, Sep. 2002, pp.
  112--121.

\bibitem{nagpal2003organizing}
R.~Nagpal, H.~Shrobe, and J.~Bachrach, ``Organizing a global coordinate system
  from local information on an ad hoc sensor network,'' in \emph{Information
  Processing in Sensor Networks}.\hskip 1em plus 0.5em minus 0.4em\relax
  Springer, 2003, pp. 333--348.

\bibitem{savvides2001dynamic}
A.~Savvides, C.-C. Han, and M.~B. Strivastava, ``Dynamic fine-grained
  localization in ad-hoc networks of sensors,'' in \emph{Proceedings of the 7th
  annual international conference on Mobile computing and networking}.\hskip
  1em plus 0.5em minus 0.4em\relax ACM, 2001, pp. 166--179.

\bibitem{ji2004sensor}
X.~Ji and H.~Zha, ``Sensor positioning in wireless ad-hoc sensor networks using
  multidimensional scaling,'' in \emph{INFOCOM 2004. Twenty-third AnnualJoint
  Conference of the IEEE Computer and Communications Societies}, vol.~4.\hskip
  1em plus 0.5em minus 0.4em\relax IEEE, 2004, pp. 2652--2661.

\bibitem{anderson2009graphical}
B.~D. Anderson, P.~N. Belhumeur, T.~Eren, D.~K. Goldenberg, A.~S. Morse,
  W.~Whiteley, and Y.~R. Yang, ``Graphical properties of easily localizable
  sensor networks,'' \emph{Wireless Networks}, vol.~15, no.~2, pp. 177--191,
  2009.

\bibitem{deghat2011distributed}
M.~Deghat, I.~Shames, B.~D. Anderson, and J.~M. Moura, ``Distributed
  localization via barycentric coordinates: Finite-time convergence,'' in
  \emph{18th World Congress of the International Federation of Automatic
  Control,(IFAC 2011), Milan, Italy}, 2011, pp. 7824--7829.

\bibitem{lederer2009connectivity}
S.~Lederer, Y.~Wang, and J.~Gao, ``Connectivity-based localization of
  large-scale sensor networks with complex shape,'' \emph{ACM Transactions on
  Sensor Networks (TOSN)}, vol.~5, no.~4, p.~31, 2009.

\bibitem{fox2003bayesian}
D.~Fox, J.~Hightower, L.~Liao, D.~Schulz, and G.~Borriello, ``Bayesian
  filtering for location estimation,'' 2003.

\bibitem{amundson2009survey}
I.~Amundson and X.~D. Koutsoukos, ``A survey on localization for mobile
  wireless sensor networks,'' in \emph{Mobile Entity Localization and Tracking
  in GPS-less Environnments}.\hskip 1em plus 0.5em minus 0.4em\relax Springer,
  2009, pp. 235--254.

\bibitem{PatwariThesiss}
N.~Patwari, ``{Location Estimation in Sensor Networks},'' Ph.{D}., University
  of Michigan--Ann Arbor, 2005.

\bibitem{camera}
Y.-G. Kim, J.~An, and K.-D. Lee, ``Localization of mobile robot based on fusion
  of artificial landmark and {RF TDOA} distance under indoor sensor network,''
  \emph{International Journal of Advanced Robotic Systems}, vol.~8, no.~4, pp.
  203--211, September 2011.

\bibitem{7438736}
G.~Han, J.~Jiang, C.~Zhang, T.~Q. Duong, M.~Guizani, and G.~K. Karagiannidis,
  ``A survey on mobile anchor node assisted localization in wireless sensor
  networks,'' \emph{IEEE Communications Surveys Tutorials}, vol.~18, no.~3, pp.
  2220--2243, 2016.

\bibitem{choset2005principles}
H.~M. Choset, \emph{Principles of robot motion: theory, algorithms, and
  implementation}, 2005.

\bibitem{4407221}
A.~Boukerche, H.~Oliveira, E.~Nakamura, and A.~Loureiro, ``Localization systems
  for wireless sensor networks,'' \emph{Wireless Communications, IEEE},
  vol.~14, no.~6, pp. 6--12, December 2007.

\bibitem{mobius1827barycentrische}
A.~F. M{\"o}bius, \emph{Der barycentrische calcul}, 1827.

\bibitem{khan2009distributed}
U.~A. Khan, S.~Kar, and J.~M. Moura, ``Distributed sensor localization in
  random environments using minimal number of anchor nodes,'' \emph{Signal
  Processing, IEEE Transactions on}, vol.~57, no.~5, pp. 2000--2016, May 2009.

\bibitem{DBLP:journals/corr/SafaviK14}
\BIBentryALTinterwordspacing
S.~Safavi and U.~A. Khan, ``Asymptotic stability of stochastic {LTV} systems
  with applications to distributed dynamic fusion,'' \emph{IEEE Transactions on
  Automatic Control}, provisionally accepted. [Online]. Available:
  \url{http://arxiv.org/abs/1412.8018}
\BIBentrySTDinterwordspacing

\bibitem{camp2002survey}
T.~Camp, J.~Boleng, and V.~Davies, ``A survey of mobility models for ad hoc
  network research,'' \emph{Wireless communications and mobile computing},
  vol.~2, no.~5, pp. 483--502, 2002.

\bibitem{wang2009sequential}
W.~Wang and Q.~Zhu, ``Sequential monte carlo localization in mobile sensor
  networks,'' \emph{Wireless Networks}, vol.~15, no.~4, pp. 481--495, 2009.

\bibitem{sippl1986cayley}
M.~J. Sippl and H.~A. Scheraga, ``Cayley-menger coordinates,''
  \emph{Proceedings of the National Academy of Sciences}, vol.~83, no.~8, pp.
  2283--2287, 1986.

\end{thebibliography}
\end{document}